\documentclass[journal,10pt]{IEEEtran}

\usepackage{amsmath,amssymb,amsthm}

\newtheorem{theorem}{Theorem}
\newtheorem{lemma}{Lemma}

\newtheorem{remark}{Remark}

\usepackage{graphicx}
\usepackage[caption=false,font=footnotesize]{subfig}  %
\usepackage{epstopdf}
\usepackage{booktabs,threeparttable}
\usepackage{array,mdwmath,mdwtab,stfloats}
\usepackage{multicol}
\usepackage{amsmath}

\usepackage{algorithm}
\usepackage{algpseudocode}
\algrenewcommand{\algorithmicrequire}{\textbf{Input:}}
\algrenewcommand{\algorithmicensure}{\textbf{Output:}}

\usepackage{url}
\usepackage{xcolor}
\usepackage[colorlinks]{hyperref}

\usepackage{array}
\makeatletter
\newcommand{\thickhline}{%
    \noalign {\ifnum 0=`}\fi \hrule height 1pt
    \futurelet \reserved@a \@xhline
}
\newcolumntype{"}{@{\hskip\tabcolsep\vrule width 1pt\hskip\tabcolsep}}
\makeatother

\begin{document}
\title{Fluid Antenna Systems: A Geometric Approach to Error Probability and Fundamental Limits}

\author{Xusheng Zhu,
            Kai-Kit Wong, \IEEEmembership{Fellow,~IEEE},
            Hao Xu,
            Han Xiao,
            Hanjiang Hong,\\
            Hyundong Shin, \emph{Fellow, IEEE}, and
            Yangyang Zhang
\vspace{-5mm}

\thanks{(\emph{Corresponding author: Kai-Kit Wong}).}
\thanks{The work of X. Zhu and K. K. Wong is supported by the Engineering and Physical Sciences Research Council (EPSRC) under grant EP/W026813/1.}

\thanks{X. Zhu, K. K. Wong, and H. Hong are with the Department of Electronic and Electrical Engineering, University College London, London, United Kingdom. K. K. Wong is also affiliated with the Department of Electronic Engineering, Kyung Hee University, Yongin-si, Gyeonggi-do 17104, Korea (e-mail: \{xusheng.zhu;kai-kit.wong;hanjiang.hong\}@ucl.ac.uk).}
\thanks{H. Xu is with the National Mobile Communications Research Laboratory, Southeast University, Nanjing 210096, China (e-mail: hao.xu@seu.edu.cn)}
\thanks{H. Xiao is School of Information and Communications Engineering, Xi'an Jiao Tong University, China (e-mail: hanxiaonuli@stu.xjtu.edu.cn).}
\thanks{H. Shin is with the Department of Electronics and Information Convergence Engineering, Kyung Hee University, Yongin-si, Gyeonggi-do 17104, Republic of Korea (e-mail: hshin@khu.ac.kr).}
\thanks{Y. Zhang is with Kuang-Chi Science Limited, Hong Kong SAR, China (e-mail: yangyang.zhang@kuang-chi.com).}
}

\maketitle
\begin{abstract}
The fluid antenna system (FAS) concept is an emerging paradigm that promotes the utilization of the feature of shape and position reconfigurability in antennas to broaden the design of wireless communication systems. This also means that spatial diversity can be exploited in an unconventional way. However, a rigorous framework for error probability analysis of FAS under realistic spatially correlated channels has been lacking. In this paper, we fill this gap by deriving a tight, closed-form asymptotic expression for the symbol error rate (SER) that establishes the fundamental scaling law linking system's SER to the channel's spatial correlation structure. A key insight of our analysis is that the achievable diversity gain is governed not by the number of antenna ports, but by the channel's effective rank. To find this critical parameter, we propose a novel dual-pronged approach. First of all, we develop a geometry-based algorithm that extracts distinct performance thresholds from the channel's eigenvalue spectrum. Second, we theoretically prove that the effective rank converges to a fundamental limit dictated solely by the antenna's normalized aperture width. We further establish the equivalence between the threshold identified by the geometric algorithm and the derived theoretical limit, providing rigorous validation for the proposed method. Our effective rank model achieves higher accuracy than existing approaches in the literature. Building on this framework, we offer a complete characterization of diversity and coding gains. The analysis leads to a definitive design insight: FAS performance improvements are fundamentally driven by enlarging the antenna's explorable aperture, which increases the effective channel rank, whereas increasing port density within a fixed aperture yields diminishing returns.
\end{abstract}


\vspace{-2mm}
\section{Introduction}
\subsection{Background}
\IEEEPARstart{T}{o mitigate} the detrimental effects of multipath fading, multiple-input multiple-output (MIMO) antenna systems have become a cornerstone of modern wireless communications \cite{foschini-1996}. By leveraging multiple antennas at both ends, MIMO can significantly enhance link reliability by coding over the spatial domain \cite{Tarokh-1998}, or boost the capacity by creating parallel spatial streams \cite{Raleigh-1998}. MIMO has also transformed into multiuser MIMO in recent years, providing enormous multiplexing gains \cite{Wong-2002,Spencer-2004,Marzetta-2010}. Latest efforts have seen MIMO combining with reconfigurable intelligent surface (RIS) technologies to engineer the propagation environment for a variety of benefits \cite{shoj2022mimo,zhu2023risa,zhu2024per,zhu2025disc}. However, the spatial diversity of MIMO is often dictated by the number of radio frequency (RF) chains due to their static, or fixed antenna configuration.


\subsection{Related Work}
To push MIMO performance beyond the spatial diversity limits imposed by RF chains, recent research has increasingly explored state-of-the-art reconfigurable antenna technologies. Early efforts in this direction can be traced back to 2004 \cite{1367557,Fazel-2008}, where reconfigurability was primarily applied to modify fixed antenna properties within conventional designs. In contrast, Wong {\em et al.}~introduced the fluid antenna system (FAS) concept \cite{Wong-ell2020,wong2022bruce}, which treats the antenna itself as a flexible, hardware-agnostic physical-layer resource. Rather than simply adapting existing elements, FAS leverages shape, position, and aperture-reconfigurability to fundamentally broaden system design and network optimization, thereby inspiring a new class of next-generation reconfigurable antennas, e.g., \cite{Liu-oe2025}. In \cite{Lu-2025}, an explanation from the electromagnetic perspective for FAS was given. Recent studies \cite{new2025at,hongicc2025,wu2024flu} have explored the various opportunities and challenges with FAS communications.

By considering dynamic antenna positioning, FAS allows fine displacements of an antenna's aperture to result in significant improvements in signal strength \cite{wong2020pems,wong2021flns}. To accurately understand the impact of spatial correlation on performance, subsequent efforts have adopted the eigenvalue-based channel model \cite{kham2023anew,new2024flin} to study the diversity order of FAS channels. In \cite{new2024an}, the diversity-multiplexing trade-off for a fluid MIMO system has been investigated. Performance analysis for FAS is usually very difficult because of the intricate relationship of how spatial correlation is linked to the performance metric. In \cite{nosa2024anew}, Ram\'{i}rez-Espinosa {\em et al.}~addressed this by proposing a block spatial-correction model that simplifies the channel model to maintain analytical tractability while being able to approximate the correlation structure reasonably well.

FAS has also found applications in multiuser environments, giving rise to the paradigm of fluid antenna multiple access (FAMA) \cite{kwonglimi202} that can accommodate hundreds of users on the same channel without the need of precoding. This impressive result, however, came with a strong assumption of extremely fast port switching of the FAS at each user. This then motivates the variants such as slow FAMA \cite{Wong-sfama2023,wong2022extrm,xu2024rev}, opportunistic FAMA \cite{wong2023oppf,Waqar-2024}, and the compact ultra-massive antenna-array (CUMA) receiver architecture \cite{Wong-cuma2024,won2023trsm}. Moreover, there have been studies showing great potential of FAS in improving non-orthogonal multiple access (NOMA) \cite{new2023flac,he2024mov,zheng2024fas} and rate-splitting multiple access (RSMA) \cite{Farshad-2025rsma} networks.

Indeed much progress has been made recently in integrating FAS into other technologies. For example, \cite{Hong-2025fasofdm,Hong-2025famaofdm} confirmed that FAS is still effective under orthogonal frequency division multiplexing (OFDM) settings. The synergy of RIS and FAS-aided users was investigated in \cite{cha2024onms,yao2025rix}. Intriguingly, the FAS concept can be integrated into the elements of RIS, meaning that each element is equipped with position reconfigurability. This has led to the fluid RIS (FRIS) system in \cite{FRmayper,han2025fl,han2025fr}, and the fluid integrated reflecting and emitting surfaces (FIRES) in \cite{Farshad-fires2025}, the latter of which aims to provide full coverage.

\subsection{Motivation and Contributions}
Though the potential of FAS is widely recognized, a major gap exists in the fundamental understanding of its error performance. Previous studies have primarily focused on outage probability, which is relatively easier to analyze and regarded as a proxy to other performance metrics. However, reliability is more accurately evaluated using the symbol error rate (SER).\footnote{When the modulation order is equal to $2$, the SER is degenerated to the bit error rate (BER).} Also, it remains unclear how the interplay between the number of ports, the physical aperture size, and spatial correlation truly dictates the achievable diversity gain. This lack of a rigorous analytical framework for error performance not only hinders the ability to accurately predict system behavior but prevents the formulation of clear and effective design principles.

To address these critical gaps, this paper provides the first comprehensive error probability analysis for FAS. The main contributions are summarized as follows:
\begin{itemize}
\item First, we derive a tight, closed-form asymptotic expression for the SER that applies to a general class of coherent modulation schemes. This result is pivotal as it establishes, for the first time, the fundamental scaling law that governs the relationship between system performance and the channel's spatial correlation structure. By providing a direct analytical link to the system's physical parameters, this expression moves beyond qualitative descriptions and enables precise, and quantitative performance evaluation in the high signal-to-noise ratio (SNR) regime.

\item We rigorously establish that the diversity of FAS is fundamentally governed not by the nominal number of ports, but by the effective rank of the channel correlation matrix. To determine this critical parameter, we propose a novel, dual-pronged approach that synergizes practical, data-driven analysis with fundamental theory. This includes a practical, geometry-based algorithm that autonomously identifies three distinct performance thresholds from the channel's eigenvalue spectrum, offering a granular view of the channel's signal, transition, and noise subspaces. Complementing this, a theoretical derivation proves that the effective rank converges to a fundamental limit determined solely by the normalized aperture width.

\item We crucially demonstrate the equivalence between the principal threshold identified by our geometric algorithm and the derived theoretical limit. This validation is significant since it indicates that the empirically observed knee in the eigenvalue spectrum is not an arbitrary artifact, but rather corresponds to a fundamental physical limit of the channel. Furthermore, by comparing against established benchmarks, we show that our effective rank model is more accurate than existing methods in the literature, offering a more precise tool for system analysis, particularly in large-aperture scenarios where other models falter.

\item Besides, based on the validated effective rank framework, we provide a complete characterization of the diversity and coding gains, decoupling the effects of the physical channel from the signaling scheme. The analysis culminates in a definitive design directive, grounded in our theoretical and simulation results: performance enhancement is fundamentally achieved by expanding the antenna's explorable aperture, which directly increases the effective channel rank and achievable diversity gain. In contrast, increasing the density of ports within a constrained space is shown to provide diminishing returns, yielding only marginal performance benefits beyond a quantifiable saturation point identified by our geometric method.
\end{itemize}

The remainder of this paper is organized as follows. Section \ref{sec:model} introduces the system model, including the channel and spatial correlation structures. In Section \ref{sec:analysis}, we derive a closed-form asymptotic SER expression. Then Section \ref{sec:valid} provides the analysis of the degree of freedom (DoF) in FAS, where we first introduce a geometry-based algorithm and then provide a theoretical analysis to validate the approach. In Section \ref{sec:result}, numerical results are presented to validate our analytical framework. Finally, Section \ref{sec:conclude} concludes this paper.

\textit{Notations}: Throughout this paper, scalars are denoted by italic letters (e.g., $N$), while vectors and matrices are represented by boldface lowercase (e.g., $\mathbf{h}$) and uppercase (e.g., $\mathbf{J}$) letters, respectively. The operators $(\cdot)^T$ and $(\cdot)^H$ denote the transpose and Hermitian transpose, respectively. $\mathrm{E}[\cdot]$ represents statistical expectation. $|\cdot|$ indicates the absolute value of a scalar or the magnitude of a complex number, and $\|\cdot\|$ denotes the Euclidean norm of a vector. Also, $\mathrm{diag}(\cdot)$ creates a diagonal matrix from a vector. $\det(\cdot)$ and $\mathrm{Rank}\{\cdot\}$ denote the determinant and rank of a matrix, respectively. $(\cdot)!!$ is the double factorial. Additionally, $J_0(\cdot)$ is the zeroth-order Bessel function of the first kind, $Q(\cdot)$ is the Gaussian Q-function, and $\Gamma(\cdot)$ is the Gamma function. The notation $\mathcal{CN}(\mu, \sigma^2)$ denotes a circularly symmetric complex Gaussian distribution with mean $\mu$ and variance $\sigma^2$. $\mathbf{I}_N$ is an $N \times N$ identity matrix. The variables, $N_{\rm eff1}, N_{\rm eff2}$, $N_{\rm eff3}$, respectively, represent the first, second, and third effective rank number thresholds obtained based on the geometric method.
$N_{\rm eff}^{\rm theo}$ denotes the threshold of the effective rank number obtained based on the analytical method.
$N_{\rm eff}$ can represent the effective rank number obtained based on either geometric or analytical methods.

\begin{figure}[t]
\centering
\includegraphics[width=.8\columnwidth]{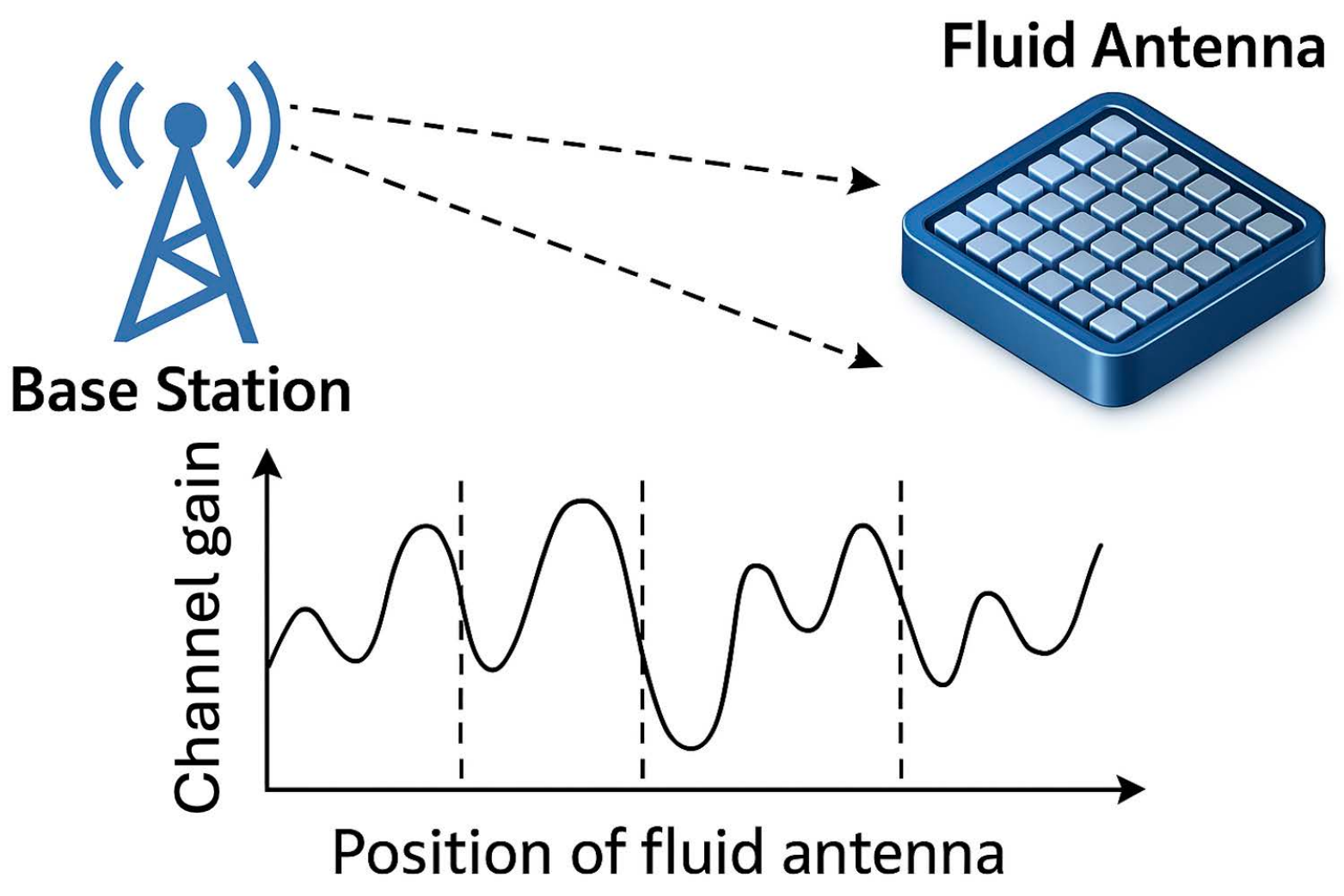}\\
\caption{Illustration of a single-user FAS model.}\label{framew}
\vspace{-2mm}
\end{figure}

\vspace{-2mm}
\section{System Model}\label{sec:model}
As shown in Fig.~\ref{framew}, we consider a point-to-point communication problem where a base station (BS) with a fixed-position antenna communicates to an FAS receiver wherein a single antenna element can be positioned at any of $N$ uniformly spaced locations, hereafter referred to as ports. These ports are distributed along a linear aperture of length $W\lambda$, where $\lambda$ is the carrier wavelength. The spatial separation between the $i$-th and $j$-th ports is therefore given by\footnote{This spatial formulation is fundamental to the subsequent analysis of spatially correlated fading, where the correlation structure is critically dependent on the inter-port distances.}
\begin{equation}\label{Deltadij}
\Delta d_{i,j} = \frac{|i - j|}{N - 1}W\lambda, \quad \text{for } i,j \in \{1, 2, \ldots, N\},
\end{equation}
where $W$ represents the normalized aperture width.
\subsection{Channel and Signal Model}
The signal received at the $n$-th port is modeled as
\begin{equation}\label{eq2}
y_n = h_n s + w_n, \quad n \in \{1, 2, \ldots, N\},
\end{equation}
where $s$ denotes the transmitted symbol, drawn from an $M$-ary constellation, e.g., phase-shift keying (PSK), quadrature amplitude modulation (QAM), or pulse amplitude modulation (PAM). The channel coefficient, $h_n = g_n e^{j\phi_n}$, is flat fading, where the amplitude $g_n$ is Rayleigh-distributed and the phase $\phi_n$ is uniformly distributed over $[-\pi, \pi)$. The term, $w_n \sim \mathcal{CN}(0, 1)$, denotes additive white Gaussian noise (AWGN).

The FAS always selects the port exhibiting the strongest channel magnitude for communication, i.e.,
\begin{equation}
g_{\rm FAS} = \max_{n \in \{1,\ldots,N\}} g_n.
\end{equation}
The instantaneous SNR at the selected port is thus given by
\begin{equation}
\gamma_{\rm FAS} = \bar{\gamma} |g_{\rm FAS}|^2,
\end{equation}
where $\bar{\gamma}$ denotes the average SNR.

\begin{table*}[t]
\centering
\caption{Parameters $p$ and $k$ for Coherent Modulation Schemes}
\label{tab:modulation_table}
\begin{tabular}{c|c|c|c}
\thickhline
\textbf{Modulation Type} & \textbf{Conditional SER} & \textbf{$p$} & \textbf{$k$} \\
\hline\hline
BPSK & $Q\left(\sqrt{2\gamma}\right)$ & $1$ & $2$ \\
\hline
$M$-PSK, $M \geq 4$ & $\approx 2Q\left(\sqrt{2\gamma} \sin\frac{\pi}{M}\right)$ & $2$ & $2\sin^2\left( \frac{\pi}{M} \right)$ \\
\hline
$M$-PAM & $2\left(1 - \frac{1}{M}\right) Q\left(\sqrt{ \frac{6\gamma}{M^2 - 1} }\right)$ & $2\left(1 - \frac{1}{M}\right)$ & $\dfrac{6}{M^2 - 1}$ \\
\hline
$M$-QAM & $4\left(1 - \frac{1}{\sqrt{M}}\right) Q\left(\sqrt{ \frac{3\gamma}{M - 1} }\right)$ & $4\left(1 - \frac{1}{\sqrt{M}}\right)$ & $\dfrac{3}{M - 1}$ \\
\thickhline
\end{tabular}
\end{table*}

\subsection{Channel Correlation Model}
The performance of FAS is fundamentally governed by the spatial correlation among the channel coefficients at different ports. To analyze this, we define the channel vector as $\mathbf{h} \triangleq [h_1, h_2, \ldots, h_N]^T$. The statistical properties of this vector are captured by its covariance matrix, $\mathbf{R} = \mathrm{E}[\mathbf{h}\mathbf{h}^H]$. This matrix can be decomposed to separate the effects of average received power and spatial correlation as
\begin{equation}\label{eq:corr_matrix}
    \mathbf{R} = \mathbf{T} \mathbf{J} \mathbf{T},
\end{equation}
where $\mathbf{T} = \mathrm{diag}(\sqrt{\bar{\gamma}_1}, \sqrt{\bar{\gamma}_2}, \ldots, \sqrt{\bar{\gamma}_N})$ is a diagonal matrix of the root-mean-square signal levels at each port, and $\mathbf{J}$ is the normalized spatial correlation matrix. The $(i,j)$-th element of $\mathbf{J}$ is the correlation coefficient given by
\begin{equation}\label{eq5j}
    J_{i,j} = \frac{\mathrm{E}[h_i h_j^*]}{\sqrt{\mathrm{E}[|h_i|^2] \mathrm{E}[|h_j|^2]}}.
\end{equation}
For this work, we adopt the widely used Jake's model to characterize the spatial correlation, which is particularly suitable for one-dimensional antenna arrays assuming rich scattering environments. The entries of $\mathbf{J}$ are thus given by
\begin{equation}
    J_{i,j}=J_0\left(\frac{2\pi\Delta d_{i,j}}{\lambda}\right)=J_0\left(2\pi W \frac{|i-j|}{N-1}\right).
\end{equation}
To represent the correlated channel, we perform an eigenvalue decomposition of $\mathbf{J}$, yielding
\begin{equation}\label{eq:eig_decomp}
    \mathbf{J} = \mathbf{U} \boldsymbol{\Lambda} \mathbf{U}^H,
\end{equation}
where $\mathbf{U}$ is an $N \times N$ unitary matrix whose columns are the eigenvectors of $\mathbf{J}$, and $\boldsymbol{\Lambda} = \mathrm{diag}(\lambda_1, \lambda_2, \ldots, \lambda_N)$ is a diagonal matrix containing the corresponding real, non-negative eigenvalues, sorted in descending order. This decomposition allows the channel vector $\mathbf{h}$ to be expressed as a linear transformation of a vector of independent and identically distributed (i.i.d.) complex Gaussian random variables
\begin{equation}
    \mathbf{h} = \mathbf{U}\boldsymbol{\Lambda}^{\frac{1}{2}} \mathbf{z},
\end{equation}
where $\mathbf{z} \sim \mathcal{CN}(\mathbf{0}, \mathbf{I}_N)$. This representation is instrumental for the performance analysis in the subsequent sections.

\vspace{-2mm}
\section{Performance Analysis}\label{sec:analysis}
Here, we derive a closed-form asymptotic expression for the average SER of the FAS. By leveraging high-SNR asymptotic techniques, we obtain a tractable SER expression that provides fundamental insights into the system's error performance.

\subsection{Conditional Error Probability}
Consider a coherent modulation with conditional SER as
\begin{equation}\label{cpep}
P_e(x)=pQ(\sqrt{kx\bar{\gamma}}),
\end{equation}
where $p$ and $k$ are constants related to the modulation format. The modulation parameters are given in Table \ref{tab:modulation_table}.

Because \eqref{cpep} contains the $Q$-function in its integral form, a closed-form expression is not attainable. This limitation obscures analytical insight and prevents a precise assessment of how each system parameter influences performance.

\subsection{Unconditional Error Probability}
The average SER is formally defined as the expectation of the instantaneous SER over all possible channel realizations characterized by their probability density function (PDF). Mathematically, the instantaneous SER can be found by
\begin{equation}\label{PExpn}
\begin{aligned}
P_{E} = p \int_{0}^{\infty} f(x)Q\left(\sqrt{{kx\overline{\gamma}}}\right) dx = p\bar{P}_e,
\end{aligned}
\end{equation}
where $\bar{P}_e = \int_{0}^{\infty} f(x)Q\left(\sqrt{kx\overline{\gamma}}\right) dx$ denotes the average SER contribution from the PDF of the random variable $x$.

To proceed, we employ Lemma~1 to obtain the high-SNR approximation of $f(x)$.

\begin{lemma}\label{lemma1}
In the high-SNR regime, the PDF of $x=\gamma_{\rm FAS}$ can be approximated as
\begin{equation}
f(x) \approx \frac{N x^{N-1}}{\det(\mathbf{J}) \prod_{n=1}^N \bar{\gamma}_n}, \quad \text{as } x \to 0^+.
\end{equation}
\end{lemma}

\begin{proof}
See Appendix~\ref{sec:appA}.
\end{proof}

Substituting the result of Lemma~\ref{lemma1} into the average SER expression yields the asymptotic form
\begin{equation}
\bar{P}_{e} \approx \int_{0}^{\infty} \frac{N x^{N-1}}{\det(\mathbf{J}) \prod_{n=1}^N \bar{\gamma}_n} Q\left(\sqrt{kx\overline{\gamma}}\right) dx.
\end{equation}
Using the integral representation of the $Q$-function \cite{zhu2022on,zhu2025toward,zhu2024on},
\begin{equation}
Q(z) = \frac{1}{\sqrt{2\pi}} \int_z^\infty e^{-v^2/2} dv,
\end{equation}
we can rewrite the expression as
\begin{equation}
\bar{P}_e \approx \frac{N}{\det(\mathbf{J}) \prod_{n=1}^N \bar{\gamma}_n \sqrt{2\pi}} \int_{0}^{\infty} x^{N-1} \int_{\sqrt{kx\overline{\gamma}}}^{\infty} e^{-v^2/2} dv  dx.
\end{equation}

By interchanging the order of integration, where the domain is defined by $0 < x < \infty$ and $\sqrt{kx\overline{\gamma}} < v < \infty$, the integration limits are transformed to $0 < v < \infty$ and $0 < x < v^2/(k\overline{\gamma})$. As a result, the expression for $P_e$ becomes
\begin{equation}\label{barpeev}
\begin{aligned}
\bar P_{e} &\approx \frac{N}{\det(\mathbf{J}) \prod_{n=1}^N \bar{\gamma}_n\sqrt{2\pi}} \int_{0}^{\infty} e^{-v^2/2} \left( \int_{0}^{\frac{v^2}{k\overline{\gamma}}} x^{N-1} dx \right) dv \\
&= \frac{N}{\det(\mathbf{J}) \prod_{n=1}^N \bar{\gamma}_n\sqrt{2\pi}} \int_{0}^{\infty} e^{-v^2/2} \left[ \frac{x^{N}}{N} \right]_0^{\frac{v^2}{k\overline{\gamma}}} dv \\
&= \frac{N}{\det(\mathbf{J}) \prod_{n=1}^N \bar{\gamma}_n\sqrt{2\pi}} \int_{0}^{\infty} e^{-v^2/2} \left( \frac{1}{N} \left( \frac{v^2}{k\overline{\gamma}} \right)^{N} \right) dv \\
&= \frac{1}{\det(\mathbf{J}) \prod_{n=1}^N \bar{\gamma}_n\sqrt{2\pi}(k\overline{\gamma})^{N}} \int_{0}^{\infty} v^{2N} e^{-v^2/2} dv.
\end{aligned}
\end{equation}
The remaining integral can be solved by relating it to the Gamma function
\begin{equation}\label{gammazd}
\Gamma(z) = \int_0^\infty y^{z-1}e^{-y}dy.
\end{equation}
Using the substitution $y = v^2/2$, which implies $v=\sqrt{2y}$ and $dv = (1/\sqrt{2y})dy$, we have
\begin{equation}
\begin{aligned}
\int_{0}^{\infty} v^{2N} e^{-v^2/2} dv & \hspace{.5mm}= \int_{0}^{\infty} (\sqrt{2y})^{2N} e^{-y} \frac{1}{\sqrt{2y}} dy \\
&\hspace{.5mm}= \int_{0}^{\infty} 2^N y^N e^{-y} 2^{-1/2} y^{-1/2} dy \\
&\hspace{.5mm}= 2^{N-1/2} \int_{0}^{\infty} y^{N-1/2} e^{-y} dy \\
&\overset{(a)}{=} 2^{N-1/2} \Gamma\left(N+\frac{1}{2}\right),
\end{aligned}
\end{equation}
where step $(a)$ results from applying (\ref{gammazd}). Accordingly, (\ref{barpeev}) can be reformulated as
\begin{equation}\label{barpeev1}
\begin{aligned}
\bar P_{e}
&\approx \frac{1}{\det(\mathbf{J}) \prod_{n=1}^N \bar{\gamma}_n\sqrt{2\pi}(k\overline{\gamma})^{N}} \left( 2^{N-1/2} \Gamma\left(N+\frac{1}{2}\right) \right).
\end{aligned}
\end{equation}
By substituting this result into (\ref{PExpn}), we obtain the final asymptotic expression as
\begin{equation}\label{Pepin}
\begin{aligned}
P_E &\approx
\frac{p2^{N-1} \Gamma(N+1/2)}{\det(\mathbf{J}) \prod_{n=1}^N \bar{\gamma}_n\sqrt{\pi}} (k\overline{\gamma})^{-N}.
\end{aligned}
\end{equation}
To obtain a closed-form expression that does not involve special functions, we further simplify the Gamma function term in (\ref{Pepin}). For half-integer arguments, the Gamma function admits the identity \cite{zhu2023ris}
\begin{equation}\label{gammahalf}
\Gamma\left(N + \frac{1}{2}\right) = \frac{(2N - 1)!!}{2^N} \sqrt{\pi},
\end{equation}
where \((2N - 1)!!\) denotes the double factorial of odd integers:
\begin{equation}
(2N - 1)!! = (2N - 1)(2N - 3)\cdots(3)(1).
\end{equation}
Substituting (\ref{gammahalf}) into (\ref{Pepin}), we obtain the fully simplified asymptotic SER as
\begin{equation}\label{PEfinal}
P_E \approx \frac{p  2^{N - 1} }{\det(\mathbf{J}) \prod_{n=1}^N \bar{\gamma}_n \sqrt{\pi}} \cdot \frac{(2N - 1)!!}{2^N} \sqrt{\pi}  (k\overline{\gamma})^{-N}.
\end{equation}
Simplifying further, we arrive at the closed-form result as
\begin{equation}\label{PEclosed}
P_E \approx \frac{p  (2N - 1)!!}{2 \det(\mathbf{J}) \prod_{n=1}^N \bar{\gamma}_n} (k\overline{\gamma})^{-N}.
\end{equation}

\begin{remark}
The asymptotic SER in (\ref{PEclosed}) provides an important analytical tool for rapidly and accurately estimating the error performance of FAS in the high-SNR regime.

{\bf Case 1:} The determinant of the correlation matrix, $\det(\mathbf{J})$, provides significant insight into the structure of the spatial channel and its impact on system performance. From matrix theory, the determinant is the product of its eigenvalues, $\prod_{n=1}^N \lambda_n$. For a normalized correlation matrix, the diagonal elements are unity, leading to the Hadamard inequality: $\det(\mathbf{J}) \leq \prod_{n=1}^N J_{n,n} = 1$. Equality holds if and only if $\mathbf{J}$ is a diagonal matrix, specifically $\mathbf{J} = \mathbf{I}_N$, which represents the ideal scenario of completely uncorrelated ports where spatial diversity is maximized.

{\bf Case 2:} Conversely, the other extreme, $\det(\mathbf{J}) = 0$, signifies a rank-deficient matrix. This condition implies linear dependence among the channel vectors at different port locations, which occurs when ports are so closely spaced within a limited aperture that their channel responses become highly correlated. In such a scenario, the system is unable to resolve distinct spatial paths, and the diversity gain is nullified.
\end{remark}

In a practical FAS, the ports are confined to a physical aperture of length $W\lambda$, which renders some degree of spatial correlation inevitable. This physical constraint implies that the system's performance is not governed by the nominal number of ports, $N$, but rather by the number of valid ports that contribute distinct spatial dimensions. Quantifying these effective DoFs is therefore crucial for an accurate performance evaluation. To this end, the subsequent section is dedicated to a detailed analysis of the number of valid ports in an FAS.

\begin{algorithm}[t]
\caption{Proposed geometry-based efficient rank identification algorithm}\label{alg:threshold_knee_descriptive}
{\small
\begin{algorithmic}[1]
    \State \textbf{Input:} Normalized aperture width $W$; Number of ports $N$.
    \State \textbf{Output:} The indices of the key turning points: ${\it N}_{\rm eff1}, {\it N}_{\rm eff2}, {\it N}_{\rm eff3}$.
    \State Construct the $N \times N$ correlation matrix $\mathbf{J}$.
    \State Compute and sort eigenvalues of $\mathbf{J}$ to obtain the spectrum $\boldsymbol{\lambda}$.
    \State Partition the spectrum $\boldsymbol{\lambda}$ into three index sets based on thresholds $\tau_1 = 10^{-10}$ and $\tau_2 = 10^{-15}$:
    \State \quad $\mathcal{R}_1 \gets \{i \mid \lambda_i > \tau_1\}$;
    \State \quad $\mathcal{R}_2 \gets \{i \mid \tau_2 < \lambda_i \le \tau_1\}$;
    \State \quad $\mathcal{R}_3 \gets \{i \mid \lambda_i \le \tau_2\}$.
    \For{each region $k \in \{1, 2, 3\}$}
        \State Let $\boldsymbol{\lambda}^{(k)}$ be the sub-spectrum corresponding to indices in $\mathcal{R}_k$.
        \State Apply geometric curvature analysis to the log-scale plot of $\boldsymbol{\lambda}^{(k)}$ to find the local index $i^*_k$ of the point with the sharpest angle.
        \State Map the local index $i^*_k$ back to its global index in $\boldsymbol{\lambda}$ to determine ${\it N}_{{\rm eff}k}$.
    \EndFor
    \State \textbf{return} ${\it N}_{\rm eff1}, {\it N}_{\rm eff2}, {\it N}_{\rm eff3}$.
\end{algorithmic}
}
\end{algorithm}

\vspace{-2mm}
\section{Valid Ports Analysis in FAS}\label{sec:valid}
In this section, we quantify the number of valid ports in an FAS. This metric is crucial as it dictates the available spatial DoF and is mathematically captured by the effective rank of the spatial correlation matrix. The effective rank serves as the foundation for our high-SNR analysis, from which we derive the system's diversity gain and coding gain.

\begin{figure*}[t]
\centering
    \subfloat[$W=0.5$]{\includegraphics[width=.24\linewidth, height=4cm]{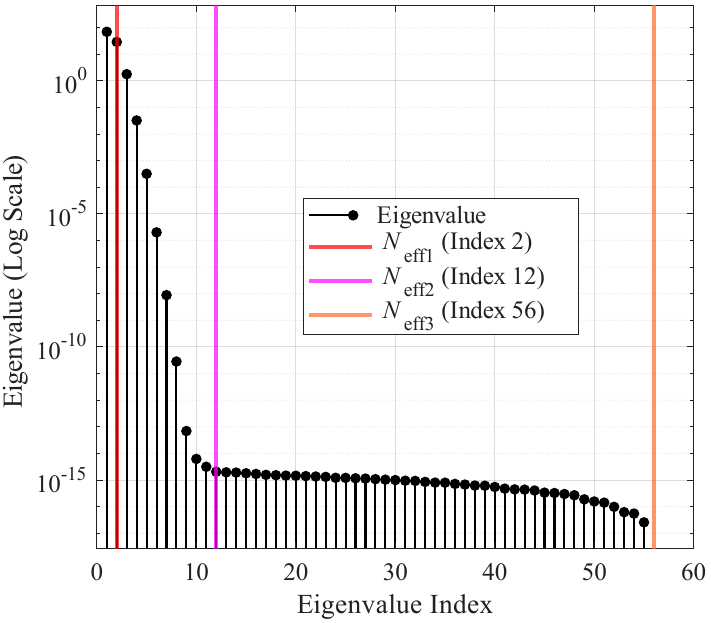}}
    \subfloat[$W=1$]{\includegraphics[width=.24\linewidth, height=4cm]{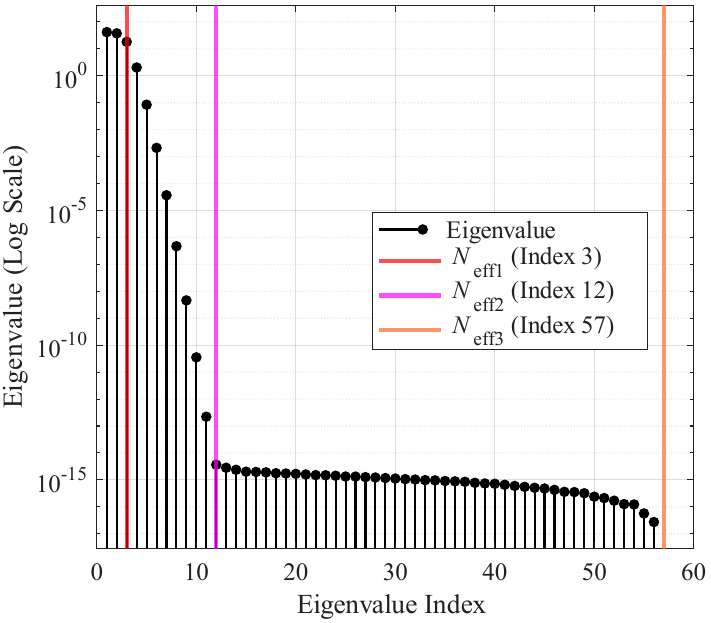}}
    \subfloat[$W=2$]{\includegraphics[width=.24\linewidth, height=4cm]{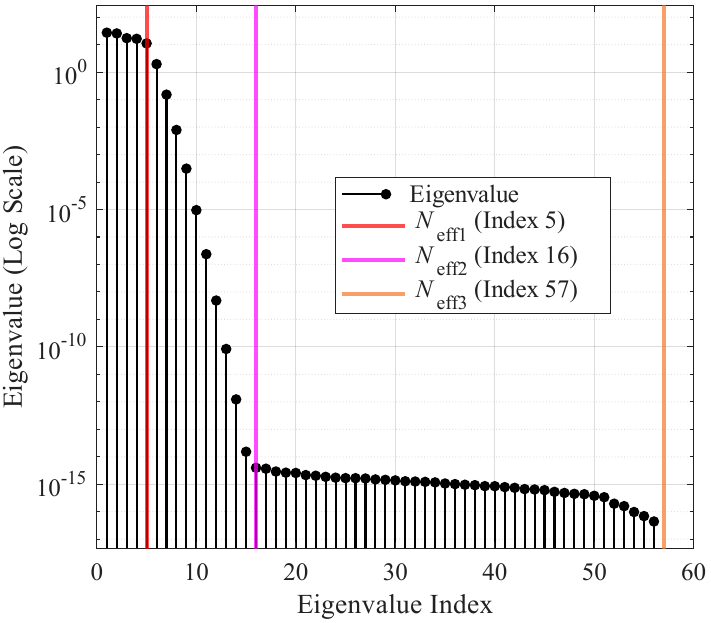}}
   \subfloat[$W=4$]{ \includegraphics[width=.24\linewidth, height=4cm]{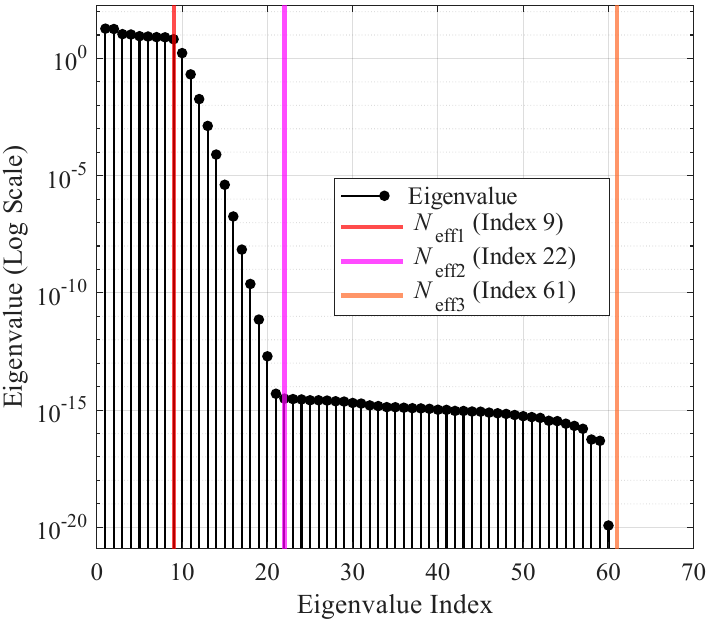}}
\caption{Eigenvalue spectrum and the three key turning points (${\it N}_{\rm eff1}, {\it N}_{\rm eff2}, {\it N}_{\rm eff3}$) identified by the proposed geometric algorithm for different aperture widths $W$ ($N=100$).}\label{fig:eigen_spectrum}
\end{figure*}

\begin{figure*}[t]
\centering
    \subfloat[$W=0.5$]{\includegraphics[width=.24\linewidth, height=4cm]{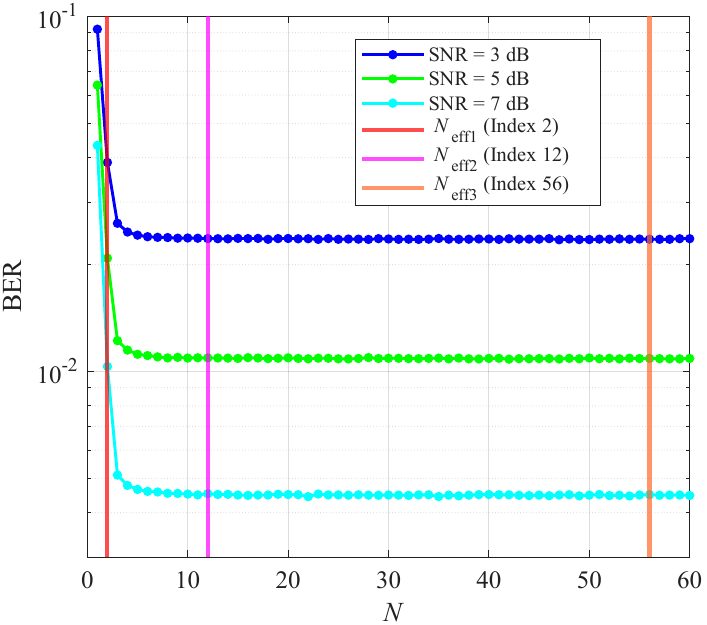}}
    \subfloat[$W=1$]{\includegraphics[width=.24\linewidth, height=4cm]{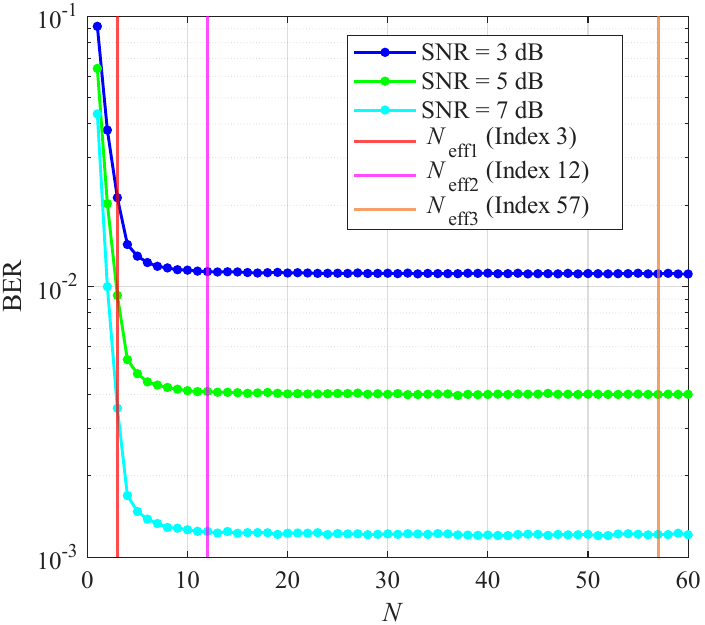}}
    \subfloat[$W=2$]{\includegraphics[width=.24\linewidth, height=4cm]{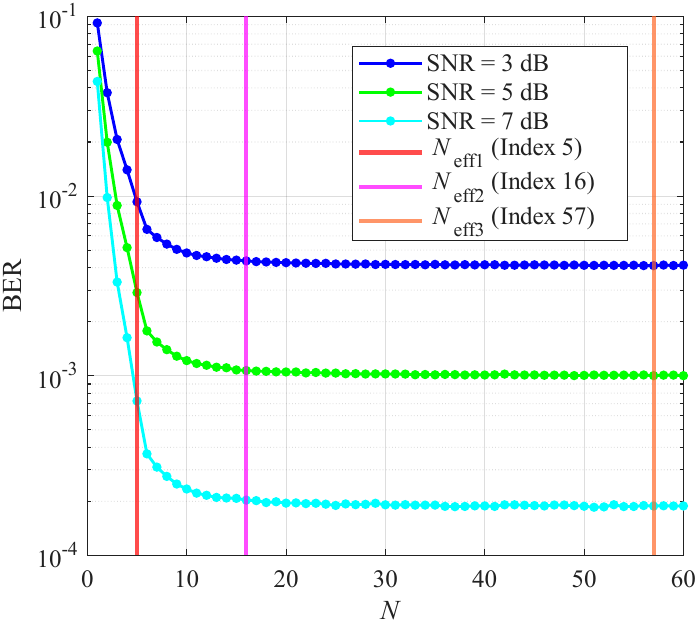}}
    \subfloat[$W=4$]{\includegraphics[width=.24\linewidth, height=4cm]{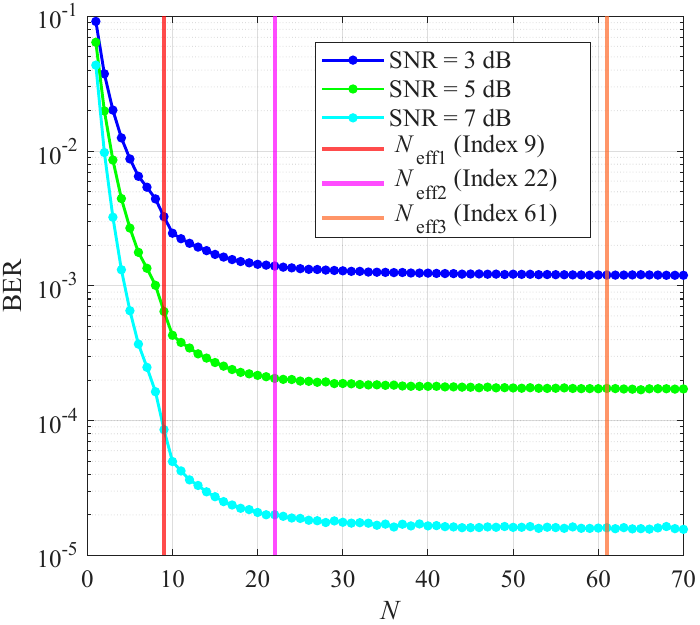}}
\caption{BER performance as a function of the number of effective rank, for different aperture widths $W$ and SNRs ($N=100$). The vertical lines correspond to the turning points identified in Fig.~\ref{fig:eigen_spectrum}.}\label{fig:ber_performance}
\end{figure*}

\subsection{Proposed Geometry-Based Method }
To provide a more granular analysis of the channel structure, we introduce a refined geometric method designed to identify not just a single effective rank, but three distinct turning points (${\it N}_{\rm eff1}, {\it N}_{\rm eff2}, {\it N}_{\rm eff3}$) that characterize the transitions between the signal-dominant, noise-dominant, and numerical-precision-limited regions of the eigenvalue spectrum.

The procedure commences with the eigenvalue decomposition of the correlation matrix $\mathbf{J}$, yielding a sorted, non-negative eigenvalue vector $\boldsymbol{\lambda}$. The core of this method is to divide the signal into three parts based on the amplitude characteristics. This division defines three different analysis areas:
\begin{enumerate}
\item The signal-dominant subspace, containing $\lambda_i > \tau_1$.
\item The transition subspace, with eigenvalues $\tau_2 < \lambda_i \le \tau_1$.
\item The noise and numerical floor, with eigenvalues $\lambda_i \le \tau_2$.
\end{enumerate}

Within each of these three regions, the geometric curvature analysis is applied independently. For each segment, a set of two-dimensional (2D) coordinates is formed from the indices and the base-10 logarithms of the corresponding eigenvalues. These coordinates are then normalized to the range $[0, 1]$ to ensure scale invariance.
The point of maximum curvature, or the knee, within each segment is identified by quantifying the angle at each interior point. For a point $\mathbf{p}_i$ in a given segment, two vectors are defined relative to its neighbors:
\begin{align}
& \mathbf{v}_1 = \mathbf{p}_{i-1} - \mathbf{p}_i,   \\
&\mathbf{v}_2 = \mathbf{p}_{i+1} - \mathbf{p}_i.
\end{align}
The cosine of the angle between these vectors is calculated using the dot product formula:
\begin{equation}
c_i = \frac{\mathbf{v}_1^T \mathbf{v}_2}{\|\mathbf{v}_1\| \|\mathbf{v}_2\|}.
\end{equation}
The local index $k_{\text{local}}$ that maximizes this cosine value corresponds to the sharpest turn within that segment. This index is then mapped back to its global index in the full spectrum $\boldsymbol{\lambda}$ to identify the turning point for that region. By performing this procedure for all three regions, we obtain the key points: ${\it N}_{\rm eff1}, {\it N}_{\rm eff2}, \text{and } {\it N}_{\rm eff3}$. This threshold-based approach provides a robust and deterministic method for characterizing the multi-stage decay profile of the channel's eigenvalue spectrum.

To validate the proposed geometric method and illustrate its practical utility, we present a joint analysis of the eigenvalue spectrum and the corresponding BER performance in Figs.~\ref{fig:eigen_spectrum} and \ref{fig:ber_performance}, respectively. Fig.~\ref{fig:eigen_spectrum} depicts the sorted eigenvalue spectrum of the correlation matrix $\mathbf{J}$ for various normalized aperture widths $W$. The vertical lines denote the three key turning points, ${\it N}_{\rm eff1}, {\it N}_{\rm eff2}, \text{and } {\it N}_{\rm eff3}$, identified by our algorithm. The first turning point, ${\it N}_{\rm eff1}$, which corresponds to the primary effective rank, scales linearly with $W$, confirming the algorithm's capability to capture the increase in channel dimensionality with aperture size. However, a more profound insight is gained by contextualizing these turning points with the empirical BER in Fig.~\ref{fig:ber_performance}. A critical observation emerges: while the BER curves exhibit diminishing returns after the first turning point (${\it N}_{\rm eff1}$), they achieve full performance saturation, indicated by a completely horizontal trajectory, precisely at the second turning point, ${\it N}_{\rm eff2}$. This direct correspondence between the spectral geometry and system performance allows us to define ${\it N}_{\rm eff2}$ as a pragmatic metric: the saturation rank, which quantitatively identifies the minimum number of ports required to exhaust all available diversity gains. The ability of our multi-stage geometric analysis to identify not only the primary effective rank (${\it N}_{\rm eff1}$) but also this practical saturation point (${\it N}_{\rm eff2}$) underscores its significant advantage over conventional single-rank estimation methods.

\subsection{Proposed Theoretical Analysis Method}

\begin{theorem}\label{thm:rank}
As $N \rightarrow \infty$, the effective rank $N_{\mathrm{eff}}^{\rm theo}$ of ${
\bf J}$ can be given by
\begin{equation}
N_{\mathrm{eff}}^{\rm theo}\triangleq 2W + 1,
\label{eq:reff_final}
\end{equation}
where the effective rank increases linearly with the normalized spatial bandwidth $W$ and is asymptotically independent of the sampling density.
\end{theorem}

\begin{proof}
See Appendix~\ref{sec:appB}.
\end{proof}

\begin{figure}[t]
\centering
\includegraphics[width=.9\columnwidth]{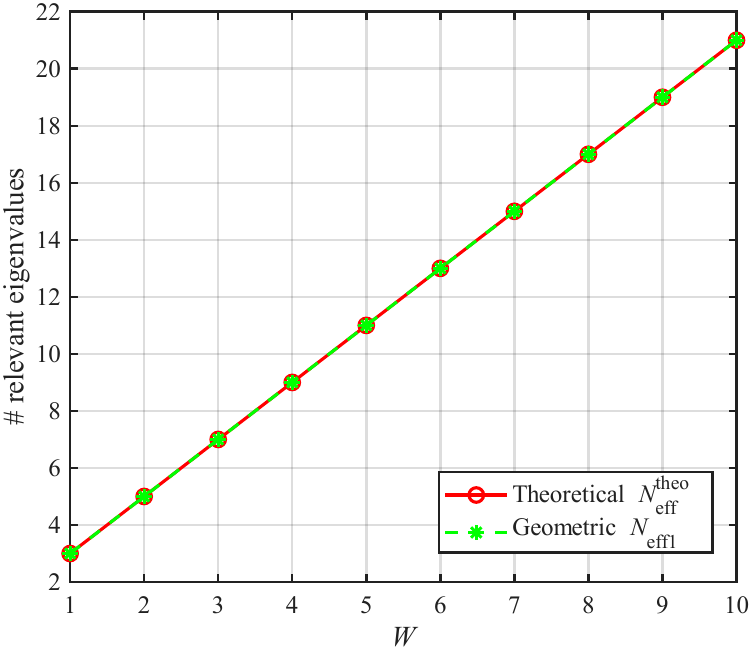}
\caption{Comparison of the effective rank estimated by the proposed geometric algorithm ($N_{\rm eff1}$) and the theoretical model ($N=40$).}\label{fig:rank_comparison}
\end{figure}

\begin{remark}
Theorem~\ref{thm:rank} provides profound insights into the design principles of FAS. It establishes the fundamental limit on the spatial DoF as $N_{\rm eff}^{\rm theo} = 2W+1$, a limit dictated by the aperture width $W$ due to the channel's finite spatial bandwidth. A key implication is the concept of port density saturation: once the number of ports $N$ is sufficient to capture these modes, further increasing $N$ provides redundant samples and yields negligible gains in effective rank. This provides a clear directive for system design, prioritizing the maximization of $W$. Finally, it is crucial to distinguish between this effective rank and spatial resolution. While a higher port density ($N > 2W+1$) does not increase the DoF, it can offer secondary benefits such as improved numerical conditioning.
\end{remark}

To validate both the theoretical framework and the proposed geometric algorithm, we present a numerical analysis in Fig.~\ref{fig:rank_comparison} and Table~\ref{tab:rank_analysis_summary}. Fig.~\ref{fig:rank_comparison} illustrates the remarkable accuracy of our geometric method, showing a near-perfect match with the theoretical $N_{\rm eff}^{\rm theo} = 2W+1$ limit across a range of $W$.

Table~\ref{tab:rank_analysis_summary} provides further quantitative validation for a high-density system ($N=100$). It confirms the linear growth of ${\it N}_{\rm eff1}$ with $W$ and highlights the principle of port saturation, noting that only a small fraction of ports (${\it N}_{\rm eff1}/N$) is required to capture the channel's full dimensionality. The most compelling evidence is the energy ratio captured by this compact subspace: the first ${\it N}_{\rm eff1}$ eigenvalues consistently contain over 97.8\% of the total channel energy. This demonstrates that our geometric algorithm not only finds a curve's turning point but also identifies the energy-dominant subspace of the channel.

\begin{table*}[]
\centering
\caption{Proposed Geometric-based Effective Rank $N_{\rm eff1}$ and Captured Energy Ratio for different $W$ ($N=100$)}
\label{tab:rank_analysis_summary}
\begin{tabular}{c||c|c|c|c|c|c}
\thickhline
\textbf{$W$} & $0.5$ & $1$ & $2$ & $4$ & $6$ & $8$ \\
\hline
\textbf{${\it N}_{\rm eff1}$} & $2$ & $3$ & $5$ & $9$ & $13$ & $17$ \\
\hline
\textbf{${\it N}_{\rm eff1}/N$} & $0.02$ & $0.03$ & $0.05$ & $0.09$ & $0.13$ & $0.17$ \\
\hline
\textbf{Energy Ratio (\%)} & $98.19$ & $97.88$ & $97.88$ & $98.08$ & $98.24$ & $98.35$ \\
\thickhline
\end{tabular}
\end{table*}
\begin{figure}[t]
  \centering
  \includegraphics[width=8cm]{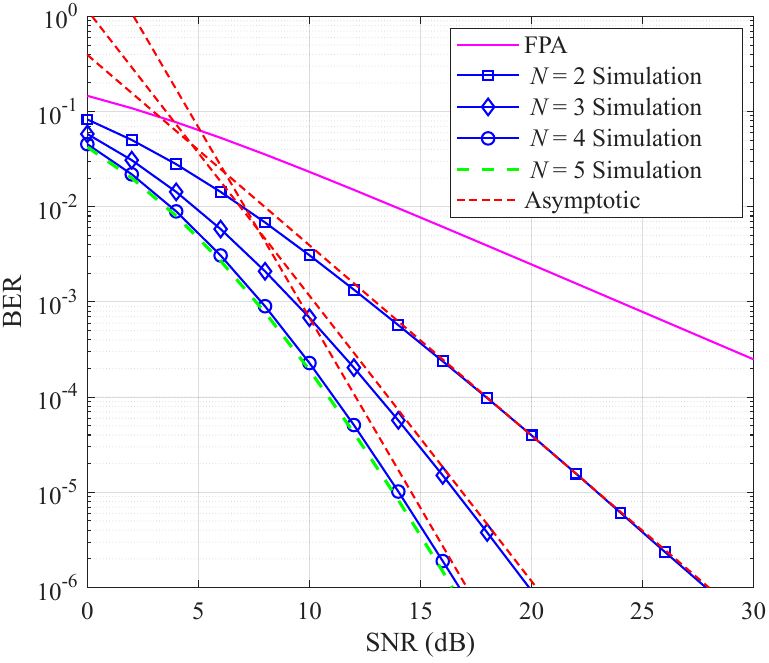}\\
  \caption{Validation of the asymptotic BER analysis for FAS with a normalized aperture $W=1$ and varying number of ports $N$. BPSK modulation is used. The solid lines represent simulation results, while the dashed lines correspond to the asymptotic theory derived in (\ref{PEclosed}). The performance of a single fixed-port antenna (FPA) is included as a benchmark.}
\label{fig:validation_N}
\end{figure}
\subsection{Diversity and Coding Gain Analysis}
Here, we provide a rigorous analysis of the diversity and coding gains, which are two pivotal performance metrics at high SNR. Our derivation is based on the asymptotic SER expression in Section \ref{sec:analysis}. This framework allows for a general analysis applicable to any FAS.

\subsubsection{Diversity Gain Derivation}
The diversity gain, denoted by $G_d$, quantifies the rate at which the error probability decreases as the SNR increases. It is formally defined as the negative slope of the asymptotic SER curve when plotted on a log-log scale. Mathematically, this is expressed as
\begin{equation}
G_{d} = -\lim_{\overline{\gamma}\rightarrow\infty}\frac{\log P_{E}(\overline{\gamma})}{\log\overline{\gamma}},
\end{equation}
where $P_E(\overline{\gamma})$ is the average SER as a function of $\overline{\gamma}$.

For correlated fading channels, the diversity order is determined by the number of independent spatial paths the channel can support, which corresponds directly to the effective rank of the channel correlation matrix, $N_{\rm eff} = {\rm Rank}\{\mathbf{J}\}$. Thus, to accurately model the system's behavior, the diversity order in the asymptotic SER expression from (\ref{PEclosed}) should be represented by $N_{\rm eff}$ instead of $N$. The resulting expression is
\begin{equation}\label{reffg}
    P_E(\overline{\gamma}) \approx \left( \frac{p(2N_{\rm eff}-1)!!}{2k^{N_{\rm eff}} \det(\mathbf{J}) \prod_{n=1}^{N_{\rm eff}}\overline{\gamma}_{n}} \right) (\overline{\gamma})^{-N_{\rm eff}}.
\end{equation}
For high-SNR analysis, we can group the terms that do not depend on $\overline{\gamma}$ into a single constant, $C$. The expression then simplifies to the form $P_E(\overline{\gamma}) \approx C \cdot (\overline{\gamma})^{-N_{\rm eff}}$.

Substituting this simplified asymptotic SER into the definition of diversity gain, we have
\begin{equation}
G_d = -\lim_{\overline{\gamma}\rightarrow\infty}\frac{\log(C \cdot \overline{\gamma}^{-N_{\rm eff}})}{\log\overline{\gamma}} = -\lim_{\overline{\gamma}\rightarrow\infty}\frac{\log C - N_{\rm eff}\log\overline{\gamma}}{\log\overline{\gamma}}.
\end{equation}
As $\overline{\gamma}\rightarrow\infty$, the term $\frac{\log C}{\log\overline{\gamma}}$ approaches zero. This leaves
\begin{equation}
G_d = -(-N_{\rm eff}) = N_{\rm eff} = {\rm Rank}\{\mathbf{J}\},
\end{equation}
where the diversity gain of FAS is equal to the effective rank of the channel correlation matrix. This result highlights that the diversity is fundamentally governed by the channel's effective DoF, not merely the number of available ports $N$.

\subsubsection{Coding Gain Derivation}
The coding gain, $G_c$, represents the system's power efficiency and corresponds to a horizontal shift of the error curve. It can be derived by expressing the asymptotic SER in the canonical form $P_E \approx (G_c \overline{\gamma})^{-G_d}$. Using our result $G_d = N_{\rm eff}$, this becomes
\begin{equation}
P_E \approx (G_c \overline{\gamma})^{-N_{\rm eff}} = G_c^{-N_{\rm eff}} (\overline{\gamma})^{-N_{\rm eff}}.
\end{equation}
We now equate this with the detailed SER expression derived from (\ref{reffg}) where $\prod_{n=1}^{N_{\rm eff}}\lambda_n$, and we assume the normalized average SNR across ports to be
\begin{equation}
P_{E}\approx\left[ \frac{p(2N_{\rm eff}-1)!!}{2 \cdot k^{N_{\rm eff}} \cdot \prod_{n=1}^{N_{\rm eff}}\lambda_n} \right] \cdot (\overline{\gamma})^{-N_{\rm eff}}.
\end{equation}
By comparing the coefficients of the $(\overline{\gamma})^{-N_{\rm eff}}$ term in both SER expressions, we establish the following identity
\begin{equation}
G_c^{-N_{\rm eff}} = \frac{p(2N_{\rm eff}-1)!!}{2 \cdot k^{N_{\rm eff}} \cdot \prod_{n=1}^{N_{\rm eff}}\lambda_n}.
\end{equation}
To solve for $G_c$, we raise both sides of the equation to the power of $(-1/r_{\rm eff})$ so that
\begin{equation}
G_c = \left( \frac{p(2N_{\rm eff}-1)!!}{2k^{N_{\rm eff}}} \right)^{-\frac{1}{N_{\rm eff}}} \cdot \left( \frac{1}{\prod_{n=1}^{N_{\rm eff}}\lambda_n} \right)^{-\frac{1}{N_{\rm eff}}}.
\end{equation}
This simplifies to the final expression for the coding gain
\begin{equation}
G_c = \left( \frac{2k^{N_{\rm eff}}}{p(2N_{\rm eff}-1)!!} \right)^{\frac{1}{N_{\rm eff}}} \left( \prod_{n=1}^{N_{\rm eff}}\lambda_n \right)^{\frac{1}{N_{\rm eff}}}.
\end{equation}
This shows that the coding gain is determined by the modulation parameters ($p, k$) and the geometric mean of the effective channel gains, which are the non-zero eigenvalues of the correlation matrix. This provides a complete and accurate characterization of the high-SNR performance.

\begin{figure*}[!t]
    \centering
    \subfloat[PSK scheme]{%
        \includegraphics[width=0.31\textwidth]{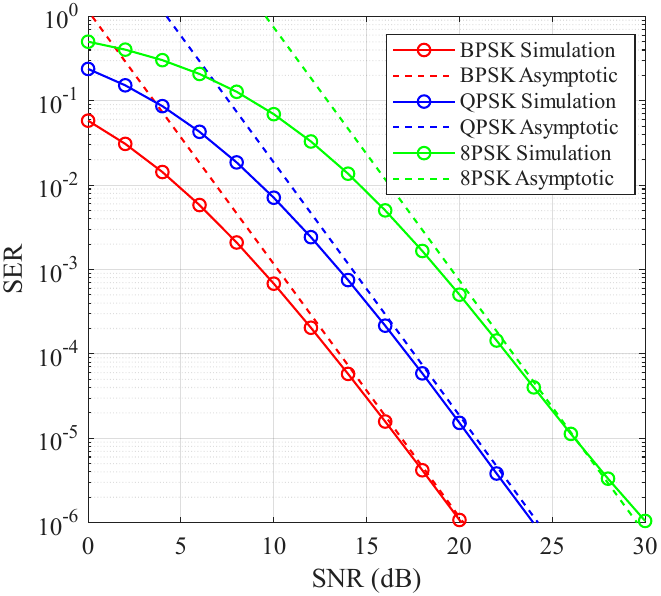}%
        \label{fig:psk}}
    \hfill
    \subfloat[QAM scheme]{%
        \includegraphics[width=0.31\textwidth]{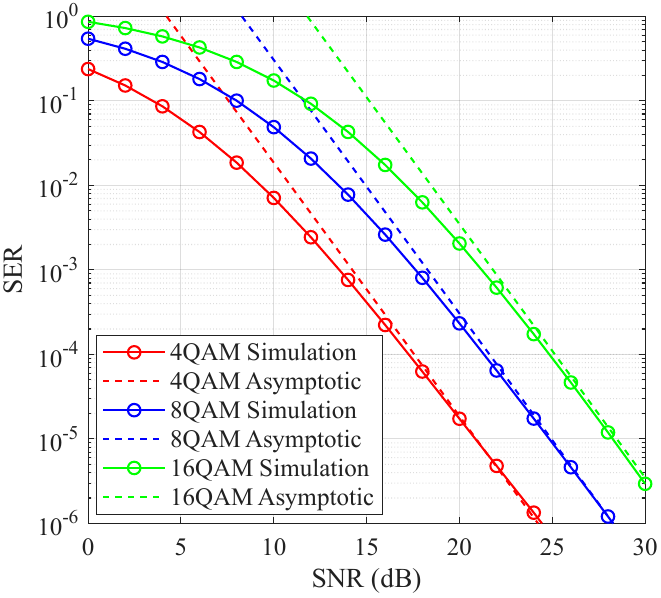}%
        \label{fig:qam}}
    \hfill
    \subfloat[PAM scheme]{%
        \includegraphics[width=0.31\textwidth]{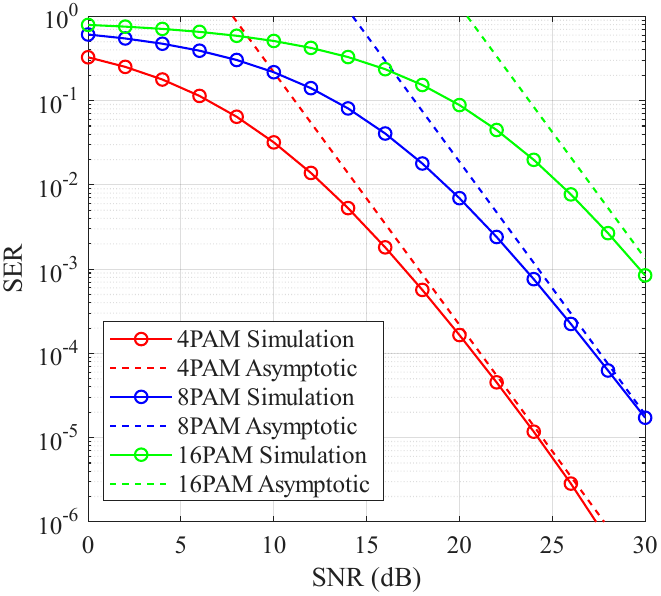}%
        \label{fig:pam}}
    \caption{SER performance comparison for various coherent modulation schemes under identical FAS parameters ($N=3$, $W=1$). The solid lines denote simulation results, and the dashed lines represent the corresponding asymptotic expressions. The parallel slopes of the curves in each subplot confirm that the diversity gain is independent of the modulation scheme.}
    \label{fig:modulation_compare}
\end{figure*}

\section{Results and Discussion}\label{sec:result}
In this section, we present numerical results to validate our theoretical framework and offer further insights into the performance of FAS. We compare Monte Carlo simulations with the derived asymptotic SER expressions to confirm their accuracy. The following analysis will focus on demonstrating the validity of the effective rank concept, quantifying the distinct effects of aperture width and port density, and verifying the robustness of our model across different modulation schemes.

Fig.~\ref{fig:validation_N} validates the accuracy of our asymptotic BER analysis by comparing it against simulations for an FAS with a fixed normalized aperture of $W=1$. The results are presented for a varying number of valid ports, $N \in \{2, 3, 4, 5\}$, using BPSK modulation. A key observation is the excellent agreement between the simulation results (solid curves) and the derived asymptotic expressions from (\ref{PEclosed}) (dashed lines), particularly in the high SNR regime. This convergence confirms the correctness of our analytical framework. Furthermore, the figure clearly illustrates the diversity gain achieved by increasing the number of available ports. As $N$ increases from $2$ to $5$, the BER curves become progressively steeper, indicating a higher diversity order. The slope of each curve in the high-SNR region directly corresponds to the diversity gain, which is equal to $N$ in these configurations, validating the result $G_d = N$ for a system where $N$ is less than or equal to the effective DoFs supported by the aperture. Even with just two valid ports ($N=2$), the FAS significantly outperforms a conventional FPA, underscoring the fundamental benefit of exploiting spatial diversity through position reconfigurability.

Fig.~\ref{fig:modulation_compare} studies the robustness and generality of our analytical framework by evaluating the SER performance for multiple coherent modulation schemes: PSK, QAM, and PAM. Here, the FAS parameters are set to $N=3$ and $W=1$. Across all subplots and for all modulation orders, the simulation results show remarkable consistency with the high-SNR asymptotic expressions. This validates the generalized SER formula in (\ref{PExpn}) and confirms the accuracy of the modulation-specific parameters $p$ and $k$ detailed in Table~\ref{tab:modulation_table}. A crucial insight from Fig.~\ref{fig:modulation_compare} is that for a fixed physical system configuration ($N, W$), the SER curves for different modulation orders are parallel to each other in the high-SNR region. This parallelism visually confirms a key theoretical finding of our work: the diversity gain ($G_d$) is determined solely by the effective rank of the channel correlation matrix, which depends on the physical parameters of the FAS, and is therefore independent of the modulation format. The horizontal shifts observed between the curves are attributed to the differences in coding gain ($G_c$), which, as derived in our analysis, is a function of the modulation parameters $p$ and $k$. This result rigorously decouples the effects of the physical channel structure (diversity gain) from the signal constellation design (coding gain).

\begin{figure}[t]
\centering
\includegraphics[width=8cm]{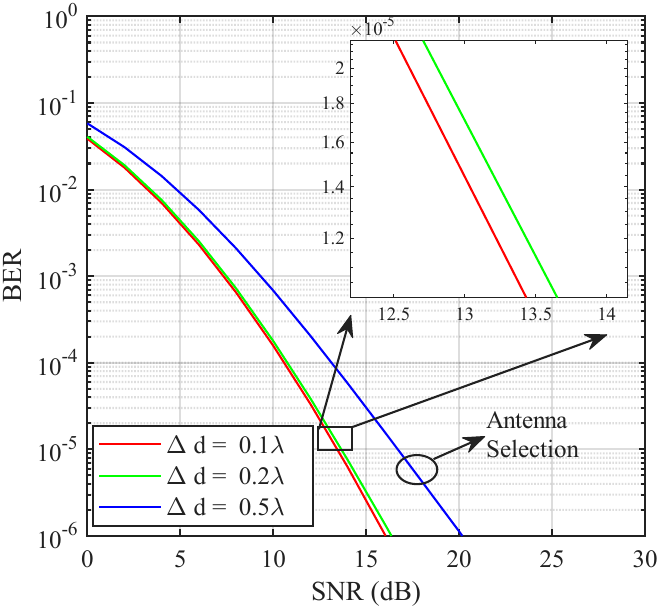}\\
\caption{Impact of port spacing $\Delta d$ on the FAS performance for a fixed aperture $W=1$. The number of ports are $N=3, 6, 11$, corresponding to spacings of $\Delta d = 0.5\lambda, 0.2\lambda, 0.1\lambda$, respectively.}\label{fig:port_spacing}
\end{figure}
\begin{figure}[t]
\centering
\includegraphics[width=.9\columnwidth]{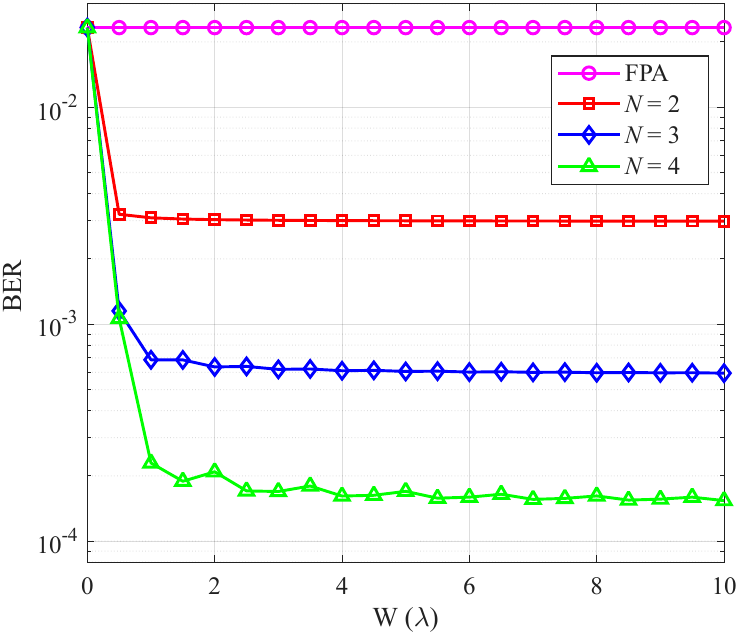}\\
\caption{Impact of the normalized aperture width $W$ on FAS performance at a fixed SNR of 10 dB for different numbers of ports $N$.}\label{fig:port_aperture}
\end{figure}

\begin{figure}[t]
\centering
\includegraphics[width=.9\columnwidth]{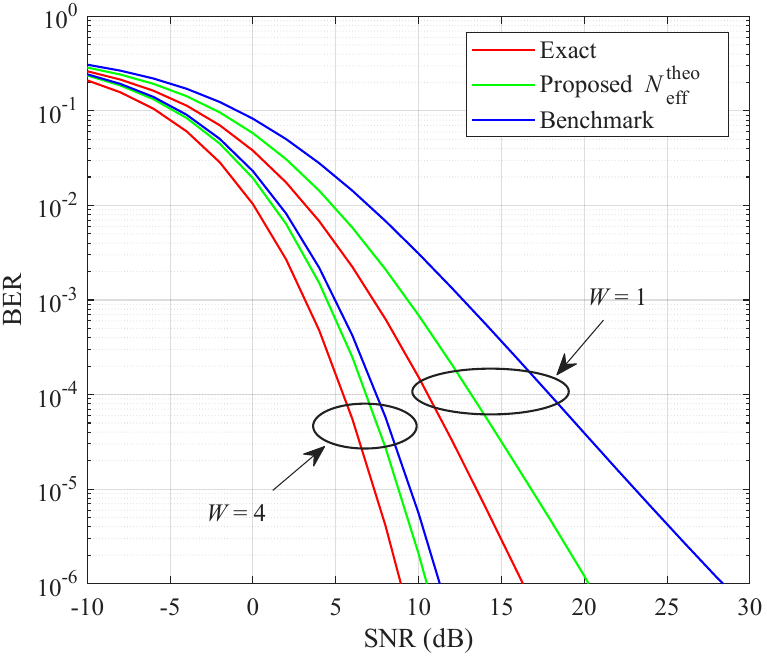}\\
\caption{Performance comparison of the proposed $N_{\rm eff}^{\rm theo}$ with the exact SER and the benchmark scheme from (19) in \cite{kham2023anew}, considering $N=100$ and two aperture widths: $W=1$ and $W=4$.}\label{fig:neff_framework}
\end{figure}

\begin{figure}[t]
\centering
\includegraphics[width=.9\columnwidth]{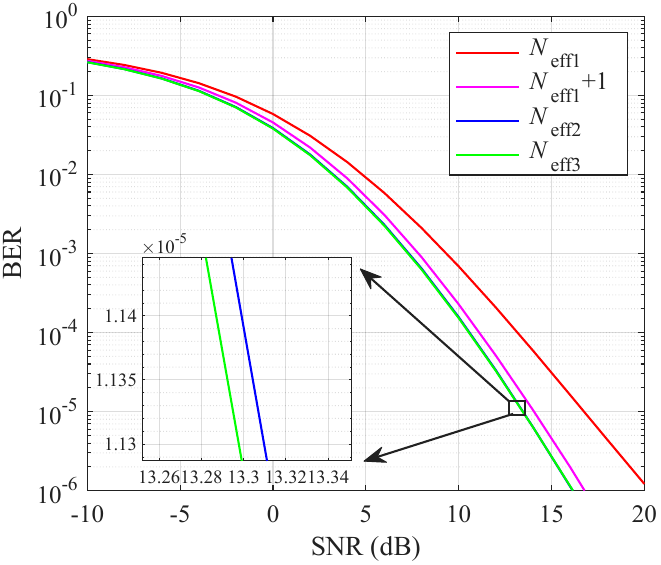}\\
\caption{Impact of the eigenvalue-based turning points on FAS performance for $W=1$. The BER is plotted for a system utilizing $N_{\rm eff1}$, $N_{\rm eff1}+1$, $N_{\rm eff2}$, and $N_{\rm eff3}$ ports, as identified by the geometric algorithm.}\label{fig:eigen_threshold}
\end{figure}

Now, Fig.~\ref{fig:port_spacing} examines the impact of port density on system performance within a fixed normalized aperture of $W=1$. We compare three configurations with varying port spacing: $\Delta d = 0.5\lambda$ ($N=3$), $\Delta d = 0.2\lambda$ ($N=6$), and $\Delta d = 0.1\lambda$ ($N=11$). The results show that decreasing the port spacing (i.e., increasing $N$) yields diminishing returns in BER. While moving from a spacing of $0.5\lambda$ to $0.2\lambda$ provides a noticeable gain, the gain from $0.2\lambda$ to $0.1\lambda$ is marginal. This behavior illustrates the concept of port density saturation, as predicted by Theorem~\ref{thm:rank}. Once the number of ports is sufficient to capture the essential spatial modes within the aperture (approximately $N \ge N_{\rm eff}^{\rm theo}=3$), adding more ports leads to higher spatial correlation and provides redundant channel information, resulting in negligible performance enhancement.\footnote{Evidently, this comment may not be valid if FAMA is concerned.}

In contrast to the limited gains from increasing port density, Fig.~\ref{fig:port_aperture} highlights the critical role of the aperture width $W$. The figure plots BER as a function of $W$ for a fixed SNR of 10 dB. For all configurations of $N$, the BER improves significantly as $W$ increases. This is because a larger aperture provides access to more diverse spatial paths, thereby increasing the channel's effective rank $N_{\rm eff}$. The performance gains, however, begin to saturate when $W$ becomes large enough such that the effective rank $N_{\rm eff}$ approaches the total number of available ports $N$. This observation reinforces a key design principle for FAS: maximizing the physical aperture is more effective for performance enhancement than merely increasing the density of ports within a constrained region.

\begin{figure*}[t]
\centering
    \subfloat{\includegraphics[width=.32\textwidth]{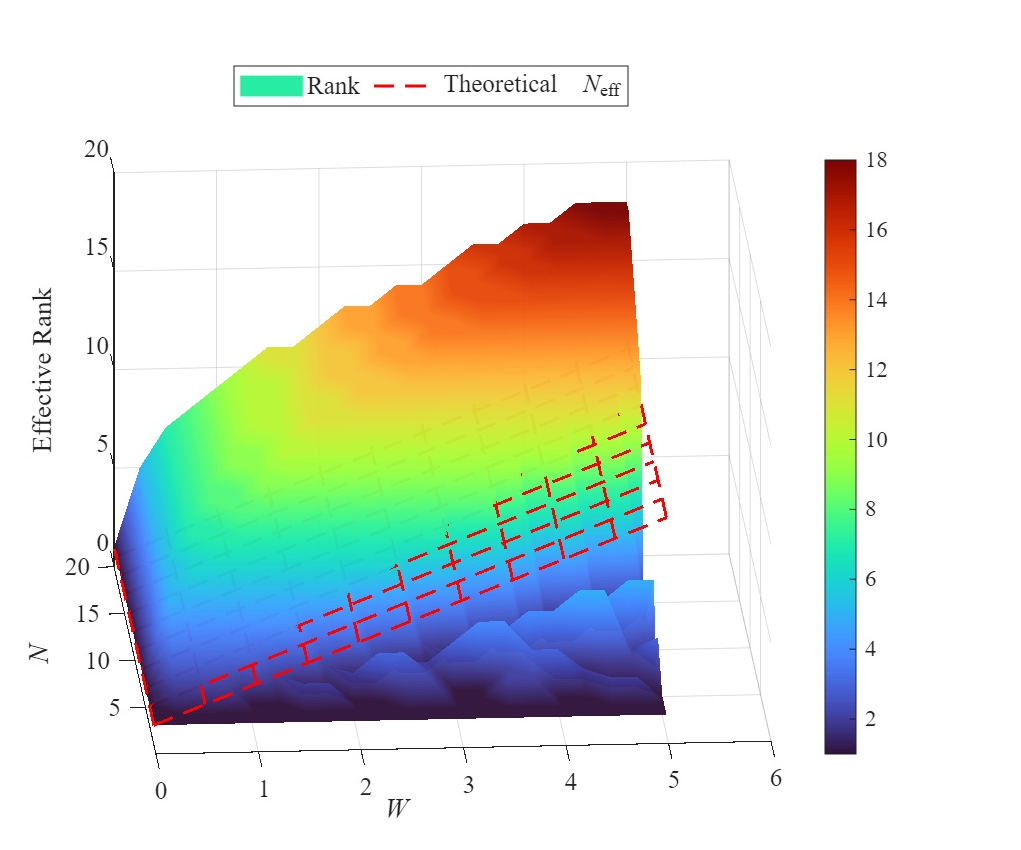}}
    \subfloat{\includegraphics[width=.32\textwidth]{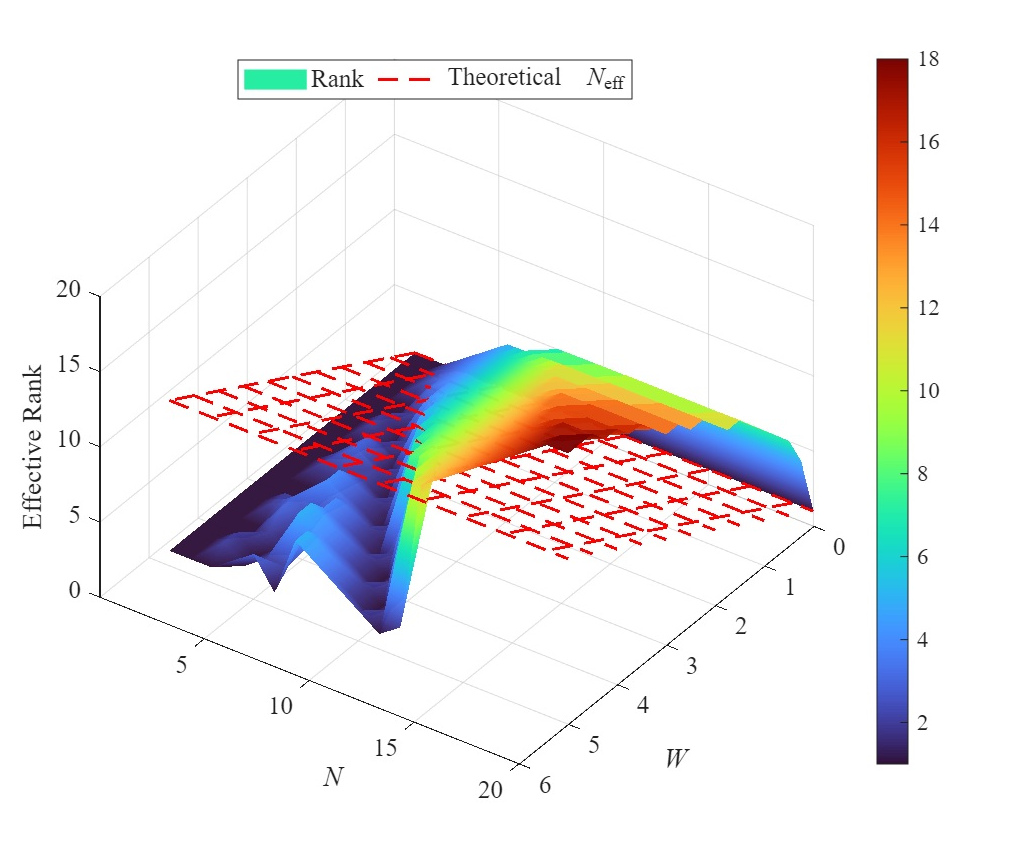}}
    \subfloat{\includegraphics[width=.32\textwidth]{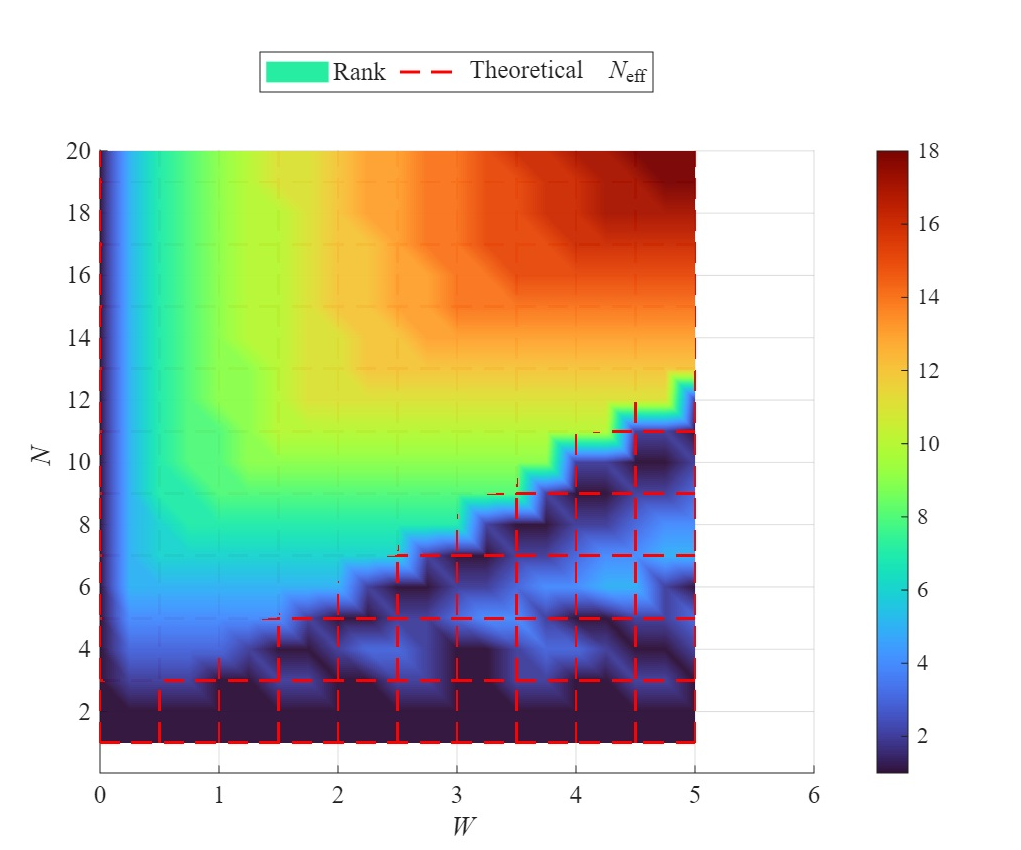}}
\caption{A comparative analysis of the effective rank estimated by the proposed geometric method versus the theoretical limit $N_{\rm eff} = 2W+1$. The surface plot shows the geometric rank as a function of $N$ and $W$, while the dashed red line indicates the theoretical value.}\label{fig:geo_vs_theory}
\end{figure*}
Fig.~\ref{fig:neff_framework} validates the proposed theoretical framework by comparing its predicted performance (Proposed theoretical $N_{\rm eff}^{\rm theo}$) against the exact BER and a benchmark scheme from \cite{kham2023anew}. The results are presented for two distinct aperture widths, $W=1$ and $W=4$. For the smaller aperture ($W=1$), our proposed framework provides a significantly tighter approximation to the exact BER curve than the benchmark model. For $W=4$, the performance gap between the two models narrows; however, our proposed framework still offers a visibly superior approximation to the exact BER. These results demonstrate that our proposed model provides a more accurate and robust performance prediction than the benchmark, particularly for smaller aperture widths where the benchmark is less precise. This validation confirms that using our effective rank in the  (\ref{eq:reff_final}) yields a more precise model for FAS performance.

Fig.~\ref{fig:eigen_threshold} provides a direct link between the key turning points identified by our geometric algorithm (Algorithm~\ref{alg:threshold_knee_descriptive}) and the actual system performance. For a system with $W=1$, we plot the BER when the number of selected ports corresponds to the identified turning points $N_{\rm eff1}$, $N_{\rm eff2}$, and $N_{\rm eff3}$. A significant performance gain is observed when moving from $N_{\rm eff1}$ ports to $N_{\rm eff1}+1$ ports. However, the performance curve for $N_{\rm eff2}$ is almost identical to that of $N_{\rm eff3}$ and shows only a marginal improvement over the $N_{\rm eff1}+1$ case. This confirms that the eigenvalues beyond the second turning point ($N_{\rm eff2}$) contribute negligibly to the overall performance. This result validates $N_{\rm eff2}$ as the ``saturation rank'', providing a practical and reliable metric for determining the minimum number of ports required to capture virtually all available diversity gain.


Finally, Fig. ~\ref{fig:geo_vs_theory} comprehensively validates the derived theoretical effective rank results by comparing the rank of our FAS port-dependent channels with the theoretical limit $N_{\rm eff}^{\rm theo}$. The surface plot shows the calculated effective rank of the FAS channel for various port numbers and aperture widths. The rank of the theoretical matrix, which is independent of $N$, is overlaid with a red dashed line. The visualization clearly shows that the rank of the FAS channel matrix closely follows the theoretical $N_{\rm eff}^{\rm theo}$ trend, especially when $N$ is sufficiently large. This remarkable consistency observed across different viewpoints confirms the validity of the theoretically derived results.
\section{Conclusions}\label{sec:conclude}
This paper presented the first rigorous error rate analysis for FAS, establishing a tight asymptotic error expression that reveals the fundamental scaling law between system performance and the port correlation matrix. Besides, we proposed a novel, dual-pronged approach to determine the channel's effective rank, which dictates its true dimensionality. In particular, this approach comprises a practical geometry-based algorithm, which uniquely identifies three distinct performance thresholds from the channel's eigenvalue spectrum, and a theoretical derivation proving that the effective rank converges to a fundamental limit determined by the physical size of the antenna's explorable space. We illustrated that this theoretical limit is equivalent to the principal threshold identified by our geometric method, thereby providing a rigorous theoretical foundation for the algorithm and validating its application. The proposed effective rank model has been demonstrated to be more accurate than existing methods in the literature. Our results have led to a key design insight: performance improvements are dictated by the expansion of the antenna's explorable space, as this directly increases the effective channel rank and achievable diversity gain. In contrast, increasing the density of potential antenna locations within a constrained space was shown to provide diminishing returns, yielding only marginal performance benefits beyond the saturation rank.

\begin{appendices}
\section{Proof of Lemma \ref{lemma1}}\label{sec:appA}
The PDF of a circularly symmetric complex Gaussian random variable $\mathbf{h}=[{h}_1,\ldots,{h}_N]^T$ can be characterized as
\begin{equation}\label{hvect}
f(\mathbf{h})=\frac{1}{\pi^N\det(\mathbf{R})}\exp\left(-\mathbf{h}^H\mathbf{Ah}\right),
\end{equation}
where $\mathbf{A}=\mathbf{R}^{-1}$. Each element of $\mathbf{A}$ can be represented in polar form as
$
A_{ij} = |A_{ij}| e^{j \phi_{A_{{ij}}}}.
$

Considering (\ref{eq2}) and the variable transformation in (\ref{hvect}), the joint PDF of $g_1, \phi_1, \dots, g_N, \phi_N$ is given by
\begin{equation}
\begin{aligned}
&f(g_1,\phi_1,\ldots,g_N,\phi_N)
\\
&
= \frac{\prod\limits_{n=1}^Ng_n\prod\limits_{n=1}^NH_n(g_n,\phi_n,\ldots,g_N,\phi_N)}{\pi^N\det(\mathbf{R})},
\end{aligned}
\end{equation}
where the term of $H_n(g_n,\phi_n,\ldots,g_N,\phi_N)$ denotes
\begin{equation}\label{Hnlar}
\begin{aligned}
&H_n(g_n,\phi_n,\ldots,g_N,\phi_N)\\=&\exp\left\{
-\mathbf{A}_{nn}g_n^2-2g_n\right.\\&\times\left.\left[
\sum_{i=n+1}^Ng_i|\mathbf{A}_{in}|\cos(\phi_n-\phi_i+\phi_{\mathbf{A}_{in}})
\right]
\right\}.
\end{aligned}
\end{equation}
To streamline notation, we denote $H_n(g_n, \phi_n, \ldots, g_N, \phi_N)$ simply as $H_n$ in the following discussion. Thus, we have
\begin{equation}
\begin{aligned}
&f(h_1, \ldots, h_N)\\=&\frac{\prod_{n=1}^Ng_n}{\pi^N\det(\mathbf{R})}
\underbrace{\int_{-\pi}^\pi\cdots\int_{-\pi}^\pi}_{N}\prod_{n=1}^NH_nd\phi_1\ldots d\phi_N.
\end{aligned}
\end{equation}
The operational principle of FAS involves activating a single port for signal reception from a large set of closely spaced candidate ports. The port switching criterion selects the port with the highest instantaneous SNR. The resulting effective SNR for the FAS, denoted as $\gamma_{\rm FAS}$, is therefore given by
\begin{equation}
\gamma_{\rm FAS}=\max(g_1^2,g_2^2,\ldots,g_N^2)=g_{\rm FAS}^2.
\end{equation}

Our strategy for determining the PDF of $\gamma_{\rm FAS}$ is to initially derive the PDF of its square root, $g_{\rm FAS}$. The target PDF is then found through a change of variables using the transformation $\gamma_{\rm FAS} = g_{\rm FAS}^2$. To tackle the PDF of $g_{\rm FAS}$, we begin with the CDF of $g_{\rm FAS}$ as
\begin{equation}
F(g_{\rm FAS})=\underbrace{\int_0^{g_{\rm FAS}}\cdots\int_0^{g_{\rm FAS}}}_{N}f(g_1,\ldots g_N)dg_1\ldots g_N.
\end{equation}
Based on this, the PDF of $g_{\rm FAS}$ can be evaluated as
\begin{equation}\label{fg_F}
\begin{aligned}
f (g_{\rm FAS})&=\frac{dF(g_{\rm FAS})}{dg_{\rm FAS}}\\
=&N\underbrace{\int_0^{g_{\rm FAS}}\cdots\int_0^{g_{\rm FAS}}}_{N-1}f(g_1,\ldots,g_{\rm FAS})dg_1\ldots dg_{N-1}\\
=&\frac{Ng_{\rm FAS}}{\pi^N\det(\mathbf{R})}\underbrace{\int_{-\pi}^{\pi}\cdots\int_{-\pi}^{\pi}}_{N}
H_NGd\phi_1\ldots\phi_N,
\end{aligned}
\end{equation}
where the term $G$ encapsulates the $N-1$ nested integrals over the channel gains $g_1,\ldots,g_{N-1}$ as
\begin{equation}
\begin{aligned}
G=&\int_0^{g_{\rm FAS}}H_{N-1}\int_0^{g_{\rm FAS}}H_{N-2}\\
&\cdots\left(\int_0^{g_{\rm FAS}}H_2\left(\int_0^{g_{\rm FAS}}H_1g_1dg_1\right)g_2dg_2\right).
\end{aligned}
\end{equation}
To evaluate this integral, we analyze its asymptotic behavior in the regime where $g_{\rm FAS}\to 0$. This allows us to leverage the following key approximations, which are derived from (\ref{Hnlar}) as
\begin{equation}
\int_0^{g_{\rm FAS}}H_n(g_n,\ldots,g_N)g_ndg_n=\frac{g_{\rm FAS}^2}{2}+o(g_{\rm FAS}^2),
\end{equation}
and
\begin{equation}\label{HN1}
H_N=1+o(1),\quad  {\rm for} \quad  g_{\rm FAS}\to 0.
\end{equation}
Starting from the innermost integral, we recursively solve the nested integrals of $G$. Applying the first approximation, we get
\begin{equation}
    \int_{0}^{g_{{\rm FAS}}}H_{1}g_{1}dg_{1} \approx \frac{g_{\rm FAS}^2}{2}.
\end{equation}
This result is then substituted into the next layer of the integral. Since the term $\frac{g_{\rm FAS}^2}{2}$ is independent of the integration variable $g_2$, it can be factored out. Applying the same approximation to the remaining integral yields
\begin{equation}
\begin{aligned}
      \int_{0}^{g_{{\rm FAS}}}H_{2}\left(\frac{g_{\rm FAS}^2}{2}\right)g_{2}dg_{2} &\approx \frac{g_{\rm FAS}^2}{2} \int_{0}^{g_{{\rm FAS}}}H_{2}g_{2}dg_{2} \\&\approx \left(\frac{g_{\rm FAS}^2}{2}\right)^2.
\end{aligned}
\end{equation}
By extrapolating this logic to the entire set of $N-1$ nested integrals, we arrive at the approximation for
\begin{equation}
\begin{aligned}
    G &= \int_{0}^{g_{\rm FAS}}H_{N-1} \cdots \left(\int_{0}^{g_{\rm FAS}}H_{1}g_{1}dg_{1}\right) \cdots g_{N-1}da_{N-1} \\&\approx \left(\frac{g_{\rm FAS}^2}{2}\right)^{N-1}.
\end{aligned}
\end{equation}
With the asymptotic approximations for $G$ and $H_N$ established, we can finalize the derivation of $f(g_{\rm FAS})$. Substituting these results into (\ref{fg_F}) yields
\begin{equation}
\begin{aligned}
    f(g_{\rm FAS}) \approx &\frac{N g_{\rm FAS}}{\pi^{N} \det(\mathbf{R})} \\&\times\underbrace{\int_{-\pi}^{\pi} \cdots \int_{-\pi}^{\pi}}_{N} (1) \left[ \left(\frac{g_{\rm FAS}^2}{2}\right)^{N-1} \right] d\phi_{1}\cdots d\phi_{N}.
\end{aligned}
\end{equation}
The integrand is now independent of the phase variables, allowing the $N$-fold integral to be readily evaluated as $(2\pi)^N$. Consequently, the expression for $f(g_{\rm FAS})$ simplifies to
\begin{equation}
\begin{aligned}
    f(g_{\rm FAS}) \approx \frac{N g_{\rm FAS}}{\pi^{N} \det(\mathbf{R})} \left(\frac{g_{\rm FAS}^2}{2}\right)^{N-1} (2\pi)^N = \frac{2N g_{\rm FAS}^{2N-1}}{\det(\mathbf{R})}.
\end{aligned}
\end{equation}
Having derived the asymptotic PDF of $g_F$, we proceed to determine the PDF of the effective SNR, $\gamma_{\rm FAS}$. This is accomplished by applying the change of variables technique to the transformation $\gamma_{\rm FAS} = g_{\rm FAS}^2$. The PDF of the transformed variable is given by
\begin{equation}
    f(\gamma_{\rm FAS}) = f(g_{\rm FAS}) \left| \frac{dg_{\rm FAS}}{d\gamma_{\rm FAS}} \right|_{g_{\rm FAS} = \sqrt{\gamma_{\rm FAS}}}.
\end{equation}
The Jacobian of this transformation is $|dg_F/d\gamma_{\rm FAS}| = 1/(2\sqrt{\gamma_{\rm FAS}})$ for $\gamma_{\rm FAS} > 0$. Substituting the expressions for $f(g_{\rm FAS})$ and the Jacobian, we arrive at the final result
\begin{equation}\label{frfR1}
\begin{aligned}
    f(\gamma_{\rm FAS}) &= \left[ \frac{2Ng_{\rm FAS}^{2N-1}}{\det(\mathbf{R})} \right]_{g_{\rm FAS} = \sqrt{\gamma_{\rm FAS}}} \cdot \frac{1}{2\sqrt{\gamma_{\rm FAS}}} \\
    &= \frac{2N(\gamma_{\rm FAS})^{N - 1/2}}{\det(\mathbf{R})} \cdot \frac{1}{2\gamma_{\rm FAS}^{1/2}} \\
    &= \frac{N\gamma_{\rm FAS}^{N-1}}{\det(\mathbf{R})}.
\end{aligned}
\end{equation}
Recalling from (\ref{eq:corr_matrix}), the PDF of (\ref{frfR1}) can then be rewritten as
\begin{equation}\label{frfR2}
f(\gamma_{\rm FAS}) = \frac{N\gamma_{\rm FAS}^{N-1}}{\det(\mathbf{J}) \prod_{n=1}^N \bar{\gamma}_n}.
\end{equation}
This expression provides the sought-after PDF for $\gamma_{\rm FAS}$ and thus concludes the proof.

\section{Proof of Theorem \ref{thm:rank}}\label{sec:appB}
To facilitate subsequent calculations, we convert the 2D correlation coefficient indices of (\ref{eq5j}) into a one-dimensional index represented as
\begin{equation}
J_k = J_0\left(2\pi W \frac{k}{N-1}\right),
\end{equation}
where $k\in[-N+1,N-1]$.

In the asymptotic regime as $N \rightarrow \infty$, the normalized lag $\tau = \frac{k}{N-1}$ becomes a continuous variable on $[-1,1]$, and the Toeplitz structure ensures that its generating function governs the spectral properties of $\mathbf{J}$, also called the symbol,
\begin{equation}
f(\theta) = \sum_{k=-\infty}^{\infty} J_k e^{-j k \theta}.
\end{equation}
For the chosen sequence $J_k$, the symbol $f(\theta)$ is the discrete-time Fourier transform (DFT) of the sampled Bessel correlation. Because $J_0(\cdot)$ oscillates while decaying only slowly, the resulting spectrum concentrates almost all of its energy within a finite angular-frequency band.

By employing the continuous approximation for large $N$, we have
\begin{equation}
J_k \approx J_0(2\pi W \tau), \qquad \tau \in [-1,1],
\end{equation}
and thus the symbol can be approximated by the integral
\begin{equation}
f(\theta) \approx (N-1)\int_{-1}^{1} J_0(2\pi W \tau) e^{-j (N-1)\tau \theta} d\tau.
\end{equation}
Noting the properties of the Bessel function and applying the Fourier transform relationship, it can be shown (see, e.g., \cite{ser1998on}) that $f(\theta)$ is approximately supported on $|\theta| \leq W$, i.e.,
\begin{equation}
f(\theta) \approx
\begin{cases}
C, & |\theta| \leq W, \\
0, & \text{otherwise},
\end{cases}
\end{equation}
where $C=\pi/W$.
This bandlimited property is a manifestation of the classical spatial bandwidth limitation.

In the large $N$ regime, Szeg\H{o}'s strong limit theorem states that the empirical distribution of the eigenvalues $\{\lambda_n\}$ of $\mathbf{J}$ converges weakly to the distribution induced by the symbol $f(\theta)$ over $[0,2\pi]$. Specifically, for any continuous function $F$,
\begin{equation}
\lim_{N \to \infty} \frac{1}{N} \sum_{n=1}^N F(\lambda_n) = \frac{1}{2\pi} \int_{0}^{2\pi} F(f(\theta)) d\theta.
\end{equation}
Hence, the normalized eigenvalue distribution is given by
\begin{equation}
p(\theta) = \frac{f(\theta)}{\frac{1}{{2\pi}}\int_{0}^{2\pi} f(\theta') {d\theta'}},
\end{equation}
and the effective rank is characterized as the exponential of the spectral entropy
\begin{equation}
N_{\mathrm{eff}}^{\rm theo} = \exp\left( - \frac{1}{2\pi}\int_{0}^{2\pi} p(\theta) \log p(\theta) d\theta\right).
\end{equation}
Given the approximate rectangular form of $f(\theta)$, we can compute the normalization as
\begin{equation}
\frac{1}{2\pi}\int_{0}^{2\pi} f(\theta) d\theta=1,
\end{equation}
so that
\begin{equation}
p(\theta)=
\begin{cases}
\displaystyle \frac{\pi}{\,W}, & |\theta|\le W,\\
0,&\text{otherwise}.
\end{cases}
\end{equation}
Because the \emph{angular}  power spectral density is rectangular, exactly \(2W+1\) discrete spatial harmonics with \(|k|\le W\) carry equal non-zero power, while all others are identically zero.  The discrete entropy can be characterized as
\begin{equation}
H_{\mathrm d}
   = -\sum_{k=-W}^{W} \frac{1}{2W+1}\log\frac{1}{2W+1}
   = \log(2W+1).
\end{equation}
Hence, the effective rank can be obtained as
\begin{equation}
N_{\mathrm{eff}}^{\rm theo} = e^{H_{\mathrm d}} = 2W + 1.
\end{equation}
At this time, the proof is completed.
\end{appendices}

\end{document}